\documentclass[10pt, twocolumn]{IEEEtran}

\usepackage{psfrag}
\usepackage{amssymb}
\usepackage{amsmath}
\usepackage{pifont}
\usepackage{cite}
\usepackage{graphics}
\usepackage{graphicx}
\usepackage{ctable}
\usepackage{epsfig}
\usepackage{subfigure}
\usepackage{amscd}
\usepackage{hyperref}
\usepackage{bm}
\usepackage{tikz}
\usepackage{algorithmic}
\usepackage{algorithm}

\newtheorem{theorem}{Theorem}[section]
\newtheorem{remark}{Remark}
\newtheorem{lemma}[theorem]{Lemma}

\newtheorem{corollary}[theorem]{Corollary}

\begin{document}

\title{Generalized Orthogonal Matching Pursuit}

\author{Jian~Wang,~\IEEEmembership{Student Member,~IEEE},
        Seokbeop~Kwon,~\IEEEmembership{Student Member,~IEEE},
        and Byonghyo~Shim,~\IEEEmembership{Senior Member,~IEEE}
\thanks{J. Wang, S. Kwon, and B. Shim are with School of Information and Communication, Korea University, Seoul, Korea (email: \{jwang,sbkwon,bshim\}@isl.korea.ac.kr).}
\thanks{Copyright (c) 2012 IEEE. Personal use of this material is permitted. However, permission to use this material for any other purposes must be obtained from the IEEE by sending a request to pubs-permissions@ieee.org}
\thanks{This work was supported by the KCC (Korea Communications Commission), Korea, under the R\&D program supervised by the KCA (Korea Communication Agency) (KCA-12-911-01-110) and the NRF grant funded by the Korea government (MEST) (No. 2011-0012525). A part of this paper was presented in Asilomar Conference on Signals, Systems \& Computers, Nov., 2011.}
}

\maketitle
\begin{abstract}
As a greedy algorithm to recover sparse signals from compressed
measurements, orthogonal matching pursuit (OMP) algorithm has
received much attention in recent years. In this paper, we introduce
an extension of the OMP for pursuing efficiency in reconstructing
sparse signals. Our approach, henceforth referred to as generalized
OMP (gOMP), is literally a generalization of the OMP in the sense
that multiple $N$ indices are identified per iteration. Owing to the
selection of multiple ``correct" indices, the gOMP algorithm is
finished with much smaller number of iterations when compared to the
OMP. We show that the gOMP can perfectly reconstruct any $K$-sparse
signals ($K > 1$), provided that the sensing matrix satisfies the RIP with
$\delta_{NK} < \frac{\sqrt{N}}{\sqrt{K} + 3 \sqrt{N}}$. We also
demonstrate by empirical simulations that the gOMP has excellent
recovery performance comparable to $\ell_1$-minimization technique
with fast processing speed and competitive computational complexity.
\end{abstract}


\IEEEpeerreviewmaketitle

\setcounter{page}{1}

\section{Introduction}


As a paradigm to acquire sparse signals at a rate significantly
below Nyquist rate, compressive sensing has attracted much attention
in recent years
\cite{donoho2001uncertainty,donoho2006sparse,candes2005decoding,candes2006robust,
candes2007sparsity, candes2008restricted,davenport2010analysis,tropp2007signal,baraniuk2008simple,raginsky2010compressed,wang2011exact,gao2011compressive,khajehnejad2011sparse}.
The goal of compressive sensing is to recover the sparse vector using
a small number of linearly transformed measurements. The process of
acquiring compressed measurements is referred to as {\it sensing}
while that of recovering the original sparse signals from compressed
measurements is called {\it reconstruction}. In the sensing
operation, $K$-sparse signal vector $\mathbf{x}$, i.e.,
$n$-dimensional vector having at most $K$ non-zero elements, is
transformed into $m$-dimensional measurements $\mathbf{y}$ via a
matrix multiplication with $\mathbf{\Phi}$.
The measurement is expressed as
\begin{equation}
\label{eq:modl} {\mathbf{y = \Phi x}}.
\end{equation}
Since $n > m$ for most of the compressive sensing scenarios, the
system in (\ref{eq:modl}) can be classified as an underdetermined
system having more unknowns than observations.
Clearly, it is in general impossible to obtain an accurate reconstruction of the
original input $\mathbf{x}$ using conventional ``inverse" transform
of $\mathbf{\Phi}$.
%
%
%
Whereas, it is now well known that with a prior information on the signal sparsity and a
condition imposed on $\mathbf{\Phi}$, $\mathbf{x}$ can be
reconstructed by solving the ${\ell_1}$-minimization problem
\cite{candes2008restricted}:
\begin{eqnarray}
\label{eq:basic2} \min_{\mathbf{x} } \|\mathbf{x}\|_{{1}}
\hspace{0.7cm} \mbox{subject to} \hspace{0.3cm}\mathbf{\Phi x = y}.
\end{eqnarray}
%
A widely used condition of $\mathbf{\Phi}$ ensuring the exact
recovery of $\mathbf{x}$ is called \textit{restricted
isometry property} (RIP) \cite{candes2005decoding}.
A sensing matrix $\mathbf{\Phi}$ is said to satisfy the RIP of order
$K$ if there exists a constant $\delta \in (0,1)$ such that
\begin{eqnarray} \label{eq:RIP}
\left( {1 - {\delta}} \right)\left\| {\mathbf{x}} \right\|_2^2 \leq
\left\| {{\mathbf{\Phi x}}} \right\|_2^2 \leq \left( {1 + {\delta}}
\right)\left\| {\mathbf{x}} \right\|_2^2
\end{eqnarray}
for any $K$-sparse vector $\mathbf{x}$ ($\left\|\mathbf{x}\right\|_0
\leq K$).
In particular, the minimum of all constants $\delta$ satisfying
(\ref{eq:RIP}) is referred to as an {\it isometry constant} $\delta_K$.
It has been shown that  $\mathbf{x}$ can be perfectly recovered by solving ${\ell_1}$-minimization problem if $\delta_{2K} < \sqrt{2} - 1$ \cite{candes2008restricted}.
While ${\ell_1}$-norm is convex and hence the problem can be solved
via linear programming (LP) technique, the complexity associated
with the LP is cubic (i.e., $O (n^3)$) in the size of the original vector to be
recovered \cite{sarvotham2006compressed} so that the complexity is burdensome for many real applications.

Recently, greedy algorithms sequentially investigating the support
of $\mathbf{x}$ have received considerable attention as cost
effective alternatives of the LP approach.
Algorithms in this category include orthogonal matching pursuit
(OMP) \cite{tropp2007signal}, regularized OMP (ROMP)\cite{needell2010signal}, stagewise OMP (StOMP) \cite{donoho2006sparse}, subspace pursuit
(SP) \cite{dai2009subspace}, and compressive sampling matching
pursuit (CoSaMP) \cite{needell2009cosamp}.
%
As a representative method in the greedy algorithm family, the OMP
has been widely used due to its simplicity and competitive
performance.
Tropp and Gilbert \cite{tropp2007signal} showed that, for a
$K$-sparse vector $\mathbf{x}$ and an $m \times n$ Gaussian
sensing matrix $\mathbf{\Phi}$, the OMP recovers $\mathbf{x}$ from
$\mathbf{y} = \mathbf{\Phi x}$ with overwhelming probability if the number of measurements follows $m \sim K \log n$.
Wakin and Davenport showed that the OMP can exactly reconstruct all
$K$-sparse vectors if ${\delta _{K + 1}} <  \frac{1}{3\sqrt K }$
\cite{davenport2010analysis} and Wang and Shim recently improved the
condition to ${\delta _{K + 1}} < \frac{1}{\sqrt{K} + 1}$
\cite{wang2012Recovery}.

The main principle behind the OMP is simple and intuitive: in each
iteration, a column of $\mathbf{\Phi}$ maximally correlated with the
residual is chosen ({\bf identification}), the index of this column
is added to the list ({\bf augmentation}), and then the vestige of
columns in the list is eliminated from the measurements, generating
a new residual used for the next iteration ({\bf residual update}).
%
%
%
%
%
Among these, computational complexity of the OMP is dominated by the
identification and the residual update steps.
In the $k$-th iteration, the identification requires a matrix-vector
multiplication so that the number of floating point operations (flops) becomes $(2m - 1)n$.
Main operation of the residual update is to compute the estimate of
$\mathbf{x}$, which is completed by obtaining the least squares (LS)
solution and the required flops is approximately $4km$.
Additionally, $2 km$ flops are required to perform the residual
update. Considering that the algorithm requires $K$ iterations, the
total number of flops of the OMP is about $2 Kmn + 3 K^2m$.
Clearly, the sparsity $K$ plays an important role in the complexity
of the OMP. 
When the signal being recovered is not very sparse, therefore, the OMP
may not be an excellent choice.

%
%
There have been some studies on the modification of the OMP, mainly on the
identification step, to improve the computational
efficiency and recovery performance. In \cite{donoho2006sparse}, a
method identifying more than one indices in each iteration was
proposed. In this approach, referred to as the StOMP, indices whose magnitude of correlation exceeds a
deliberately designed threshold are chosen.
It is shown that while achieving performance comparable to
$\ell_1$-minimization technique, the StOMP runs much faster than the OMP as well
as $\ell_1$-minimization technique \cite{donoho2006sparse}.
%
%
In \cite{needell2010signal}, another variation of the OMP, so called
ROMP, was proposed. After choosing a set of $K$ indices with largest
correlation in magnitude, the ROMP narrows down the candidates by
selecting a subset satisfying a predefined regularization rule.
It is shown that the ROMP algorithm exactly recovers $K$-sparse
signals under $\delta_{2K} < 0.03/\sqrt{\log K}$
\cite{needell2009uniform}.
While the main focus of the StOMP and ROMP algorithm is on the modification of the identification step, the SP and CoSaMP algorithm require additional operation, called pruning step, to refine the signal estimate recursively.

Our approach lies on the similar ground of these approaches in the
sense that we pursue reduction in complexity through the modification on the identification step of the OMP. Specifically, towards the reduction of complexity and speeding-up the
execution time of the algorithm, we choose multiple indices in each
iteration. While previous efforts employ special treatment on the
identification step such as thresholding\cite{donoho2006sparse} (for StOMP) or
regularization\cite{needell2010signal} (for ROMP), the proposed method pursues
direct extension of the OMP by choosing indices corresponding to $N$
($\geq 1$) largest correlation in magnitude.
Therefore, our approach, henceforth referred to as {\it generalized OMP}
(gOMP), is literally  a generalization of the OMP and embraces
the OMP as a special case ($N = 1$).
Owing to the selection of multiple indices, multiple ``correct"
indices (i.e., indices in the support set) are added to the list and the algorithm is finished with much smaller number of iterations when compared to the OMP. Indeed, in both empirical simulations and complexity analysis, we observe that the gOMP achieves substantial reduction in the number of calculations with competitive reconstruction performance.

The primary contributions of this paper are twofold:
\begin{itemize}
\item We present an extension of the OMP, termed gOMP, for pursuing efficiency in reconstructing
sparse signals.  
Our empirical simulation shows that the recovery performance of
the gOMP is comparable to the LP technique as well as modified OMP algorithms (e.g., CoSaMP and
StOMP).
\item We develop a perfect recovery condition of the gOMP. To be specific, we show that the RIP of order $NK$ with $\delta_{NK} < \frac{\sqrt{N}}{\sqrt{K} + 3 \sqrt{N}}$ ($K > 1$) is sufficient for the gOMP to exactly recover any $K$-sparse vector within $K$ iterations (Theorem \ref{thm:summary_gOMPgood}).
As a special case of the gOMP, we show that a sufficient
condition of the OMP is given by $\delta_{K+1} <  \frac{1}{\sqrt{K} +
1}$. Also, we extend our work to the reconstruction of sparse signals in the presence of noise and obtain the bound of the estimation error.
\end{itemize}

It has been brought to our attention that in parallel to our effort,
orthogonal multi matching pursuit (OMMP) or orthogonal super greedy
algorithm (OSGA) \cite{liu2012orthogonal} suggested a similar
treatment to the one posted in this paper. Similar approach has also
been introduced in \cite{maleh2011improved}. Nevertheless, our work
is sufficiently distinct from these works in the sufficient recovery
condition analysis. Further, we also provide an analysis of the
noisy scenario (upper bound of the recovery distortion in
$\ell_2$-norm) for which there is no counterpart in the OMMP and
OSGA study.

The rest of this paper is organized as follows.
In Section II, we introduce the proposed gOMP algorithm and provide empirical experiments on the reconstruction performance. In Section III, we provide the RIP based analysis of the gOMP guaranteeing the perfect reconstruction of $K$-sparse signals. We also revisit the OMP algorithm as a special case of the gOMP and obtain a sufficient condition ensuring the recovery of $K$-sparse signals.
In Section IV, we study the reconstruction performance of the gOMP under noisy measurement scenario. In Section V, we discuss complexity of the gOMP algorithm and conclude the paper in Section VI.

We briefly summarize notations used in this paper.
Let  $\Omega = \{1,2,\cdots,n\}$ then $T = \{\hspace{1mm}i
\hspace{1mm} |\hspace{1mm} i \in \Omega,  x_i\neq 0\}$ denotes the
support of vector $\mathbf{x}$.
For $D \subseteq \Omega$, $|D|$ is the cardinality of $D$.
$T - D = T \backslash \left( T \cap D \right)$ is the set of all
elements contained in $T$ but not in $D$.
${{\mathbf{x}}_D} \in \mathbb{R}^{|D|}$ is a restriction of the
vector ${\mathbf{x}}$ to the elements with indices in $D$.
${{\mathbf{\Phi}}_D} \in {\mathbb{R}^{m \times \left| D \right|}}$ is a
submatrix of ${\mathbf{\Phi}}$ that only contains columns indexed by
$D$.
If $\mathbf{\Phi}_D$ is full column rank, then $\mathbf{\Phi}_D^\dagger =
(\mathbf{\Phi}'_D\mathbf{\Phi}_D)^{-1}\mathbf{\Phi}'_D$ is the pseudoinverse
of $\mathbf{\Phi}_D$.
$span(\mathbf{\Phi}_D)$ is the span of columns in $\mathbf{\Phi}_D$.
$\mathbf{P}_{D}=\mathbf{\Phi}_{D} \mathbf{\Phi}_{D}^\dagger$ is the
projection onto $span(\mathbf{\Phi}_{D})$.
$\mathbf{P}_{D}^\bot = \mathbf{I}-\mathbf{P}_{D}$ is the projection
onto the orthogonal complement of $span(\mathbf{\Phi}_{D})$.

\section{gOMP Algorithm}\label{sec:algorithm}
%
In each iteration of the gOMP algorithm, correlations between
columns of $\mathbf{\Phi}$ and the modified measurements (residual)
are compared and indices of the columns corresponding to $N$ maximal
correlation are chosen as the new elements of the estimated support
set $\Lambda^k$. As a trivial case, when $N = 1$, gOMP returns to the OMP.
Denoting the $N$ indices as $\phi (1), \cdots, \phi (N)$ where
$\phi(i) = \arg \mathop {\max} \limits_{j:j \in \Omega \backslash
\{\phi(i-1),\cdots,\phi(1)\}
}|\langle\mathbf{r}^{k-1},\mathbf{\varphi}_{j}\rangle|$, the
extended support set at the $k$-th iteration becomes $\Lambda^k =
\Lambda^{k-1}\cup \{\phi (1),\cdots,\phi (N)\}.$
After obtaining the LS solution ${{\mathbf{\hat x}}_{\Lambda^{k}}} =
\arg \mathop
{\min}\limits_{\mathbf{u}}{\left\|\mathbf{y}-\mathbf{\Phi}_{\Lambda^{k}}\mathbf{u}\right\|}_{2}
= \mathbf{\Phi}_{\Lambda^{k}}^\dag \mathbf{y}$, the residual
${{\mathbf{r}}^k}$ is updated by subtracting
$\mathbf{\Phi}_{\Lambda^{k}} {{\mathbf{\hat x}}_{\Lambda^{k}}}$ from
the measurements $\mathbf{y}$.
In other words, the projection of $\mathbf{ y}$ onto the orthogonal
complement space of $span(\mathbf{\Phi}_{\Lambda^k} )$ becomes the
new residual (i.e., ${{\mathbf{r}}^k} =
\mathbf{P}_{\Lambda^{k}}^\bot \mathbf{y}$).
These operations are repeated until either the iteration number
reaches maximum $k_{\max} = \min(K, \frac{m}{N})$ or the
$\ell_2$-norm of the residual falls below a prespecified threshold
($\| \mathbf{r}^k \|_2 \leq \epsilon$).


%
%
%
%
%
%

It is worth mentioning that the residual $\mathbf{r}^k$ of the gOMP
is orthogonal to the columns of $\mathbf{\Phi}_{\Lambda^k}$ since
\begin{eqnarray}
\left\langle {{{\mathbf{\Phi }}_{{\Lambda ^k}}},{{\mathbf{r}}^k}} \right\rangle  &=& \left\langle {{{\mathbf{\Phi }}_{{\Lambda ^k}}},{\mathbf{P}}_{{\Lambda ^k}}^ \bot {\mathbf{y}}} \right\rangle   \\
& = & {\mathbf{\Phi }}_{{\Lambda ^k}}'{\mathbf{P}}_{{\Lambda ^k}}^ \bot {\mathbf{y}} \\
& = & \label{eq:brk2} {\mathbf{\Phi }}_{{\Lambda ^k}}'{\left(
{{\mathbf{P}}_{{\Lambda ^k}}^ \bot }
  \right)'}{\mathbf{y}}  \\
&=& \label{eq:brk33} {\left( {{\mathbf{P}}_{{\Lambda ^k}}^ \bot
{{\mathbf{\Phi }}_{{\Lambda ^k}}}} \right)'}{\mathbf{y}} =
{\mathbf{0}}
\end{eqnarray}
where (\ref{eq:brk2}) follows from the symmetry of ${{\mathbf{P}}_{{\Lambda ^k}}^ \bot }$ (${{\mathbf{P}}_{{\Lambda^k}}^ \bot } = \left({{\mathbf{P}}_{{\Lambda ^k}}^ \bot }\right)')$ and (\ref{eq:brk33}) is due to $${{\mathbf{P}}_{{\Lambda ^k}}^ \bot
{{\mathbf{\Phi }}_{{\Lambda ^k}}}} = \left(
\mathbf{I}-\mathbf{P}_{\Lambda ^k} \right) \mathbf{\Phi }_{\Lambda
^k} = \mathbf{\Phi}_{\Lambda ^k} - \mathbf{\Phi}_{\Lambda ^k}
\mathbf{\Phi}_{\Lambda ^k}^\dagger \mathbf{\Phi}_{\Lambda ^k} =
\mathbf{0}.$$  Here we note that this property is satisfied when
$\mathbf{\Phi}_{\Lambda^k}$ has full column rank, which is true
if $k \leq {m/N}$ in the gOMP operation.
It is clear from this observation that indices in ${\Lambda^k}$
cannot be re-selected in the succeeding iterations and the
cardinality of ${\Lambda^k}$ becomes simply $kN$.
When the iteration loop of the gOMP is finished, therefore, it is
possible that the final support set $\Lambda^s$ contains indices not
in $T$. Note that, even in this situation, the final result is
unaffected and the original signal is recovered because
\begin{eqnarray}
\hat{\mathbf{x}}_{\Lambda^s} 
&=& \mathbf{\Phi}_{{{\Lambda }^{s}}}^{\dagger }\mathbf{y}  \\
&=& \label{eq:bff1} {{\left( \mathbf{\Phi }_{{{\Lambda
}^s}}'{{\mathbf{\Phi }}_{{{\Lambda }^s}}}
\right)}^{-1}}\mathbf{\Phi }_{{{\Lambda }^s}}'{{\mathbf{\Phi
}}_{T}}{{\mathbf{x}}_{T}}     \\
&=& \label{eq:bff2} {{\left( \mathbf{\Phi }_{{{\Lambda
}^s}}'{{\mathbf{\Phi }}_{{{\Lambda }^s}}} \right)}^{-1}}\mathbf{\Phi }_{{{\Lambda }^s}}' \left( {\mathbf{\Phi}_{{{\Lambda}^s}}}{{\mathbf{x}}_{{\Lambda }^s}} \right)   \nonumber     \\
& & -{{\left( \mathbf{\Phi }_{{{\Lambda
}^s}}'{{\mathbf{\Phi }}_{{{\Lambda }^s}}} \right)}^{-1}}\mathbf{\Phi }_{{{\Lambda }^s}}' {\mathbf{\Phi}_{{{\Lambda}^s - T}}}{{\mathbf{x}}_{{\Lambda }^s - T}}   \\
 &=& \label{eq:bff3} {\mathbf{x}}_{{\Lambda }^s},
\end{eqnarray}
where (\ref{eq:bff2}) follows from the fact that ${\mathbf{x}}_{{\Lambda }^s - T } = \mathbf{0}$.
From this observation, we deduce that as long as at least one
correct index is found in each iteration of the gOMP, we can ensure
that the original signal is perfectly recovered within $K$
iterations. In practice, however, the number of correct indices
being selected is usually more than one so that the required number of
iterations is much smaller than $K$.

%
%
%
%

\begin{table*}
\begin{center}
\caption{The gOMP Algorithm} \label{tab:gOMP} 
\begin{tabular}{ll}  \hline
\textbf{Input}:  &measurements $\mathbf{y} \in \mathbb{R}^{m}$, \\
        &sensing matrix $\mathbf{\Phi} \in \mathbb{R}^{m \times n}$, \\
        &sparsity $K$,\\
        &number of indices for each selection $N$ ($N \leq K$ and $N \leq m/K$).\\
\textbf{Initialize}: &iteration count $k = 0$, \\
            &residual vector $\mathbf{r}^{0} = \mathbf{y}$, \\
            &estimated support set $\Lambda^{0} = \emptyset$. \\  \hline \\
\textbf{While}  &$\| \mathbf{r}^{k} \|_2 > \epsilon$ and $ k < \min \{K,m/N \}$ \textbf{do} \\
&$k = k + 1$.\\
            & (\textbf{Identification})~~~~~ Select indices $\{\phi(i)\}_{i = 1,2,\cdots,N}$  corresponding \\
            & \hspace{24mm}to $N$ largest entries (in magnitude) in $\mathbf{\Phi}' \mathbf{r}^{k-1}$.\\
            & (\textbf{Augmentation}) ~~~~$\Lambda^k = \Lambda^{k-1}\cup \{\phi(1),\cdots,\phi(N)\}$.\\
            & (\textbf{Estimation}) ~~~~~~~ ${{\mathbf{\hat x}}_{\Lambda^{k}}} = \arg \mathop {\min}\limits_{\mathbf{u}}{\left\|\mathbf{y}-\mathbf{\Phi}_{\Lambda^{k}}\mathbf{u}\right\|}_{2}$. \\
            & (\textbf{Residual Update}) ~${{\mathbf{r}}^k} = {\mathbf{y}} -{{\mathbf{\Phi }}_{\Lambda^{k}}}{{\mathbf{\hat{x}}_{\Lambda^{k}}}} $.\\
\textbf{End}  \\
\textbf{Output}:  & the estimated signal   $\hat{\mathbf{ x }} =
\arg \mathop{\min}\limits_{\mathbf{u}:{\text{supp}}\left(
{\mathbf{u}} \right) = \Lambda^{k}}{\left\| {{\mathbf{y - \Phi u}}}
\right\|_2}$.\\  \hline
\end{tabular}
\end{center}
\end{table*}


In order to observe the empirical performance of the gOMP algorithm, we performed computer simulations.
In our experiment, we use the testing strategy in
\cite{candes2005error,dai2009subspace} which measures the
effectiveness of recovery algorithms by checking the empirical
frequency of exact reconstruction in the noiseless environment. By
comparing the maximal sparsity level of the underlying sparse
signals at which the perfect recovery is ensured (this point is
often called {\it critical sparsity} \cite{dai2009subspace}),
accuracy of the reconstruction algorithms can be compared
empirically.
In our simulation, the following algorithms are considered.
\begin{enumerate}
\item LP technique for solving $\ell_1$-minimization problem (http://cvxr.com/cvx/).
\item OMP algorithm.
\item gOMP algorithm.
\item StOMP with false alarm control (FAC) based thresholding (http://sparselab.stanford.edu/).\footnote{Since FAC scheme outperforms false discovery control (FDC) scheme, we exclusively use FAC scheme in our simulation.}
    %
\item ROMP algorithm    \\
(http://www.cmc.edu/pages/faculty/DNeedell).
\item CoSaMP algorithm  \\
(http://www.cmc.edu/pages/faculty/DNeedell).
\end{enumerate}

\begin{figure}[t]
\begin{center}
\includegraphics[width=95mm]{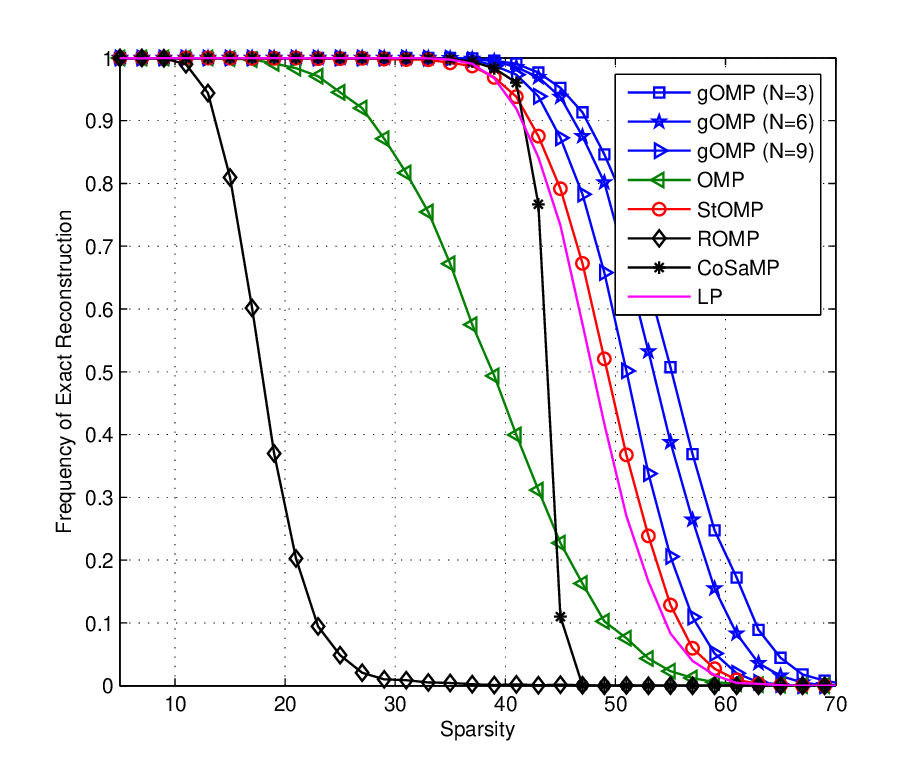}
 \caption{Reconstruction performance for $K$-sparse Gaussian signal
vector as a function of sparsity $K$.}\label{fig:GaussER}
\end{center}
\end{figure}

\begin{figure}[t]
\begin{center}
\includegraphics[width=95mm]{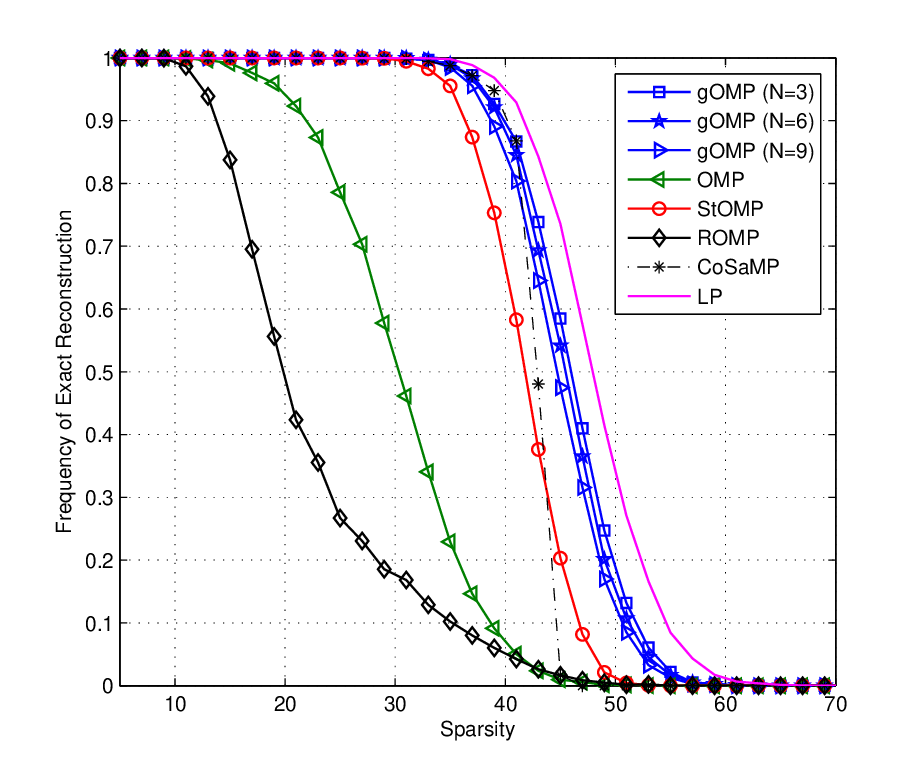}
 \caption{Reconstruction performance for $K$-sparse PAM signal
vector as a function of sparsity $K$.} \label{fig:PAMER}
\end{center}
\end{figure}

In each trial, we construct $m\times n$ ($m = 128$ and $n = 256$)
sensing matrix $\mathbf{\Phi}$ with entries drawn independently from
Gaussian distribution $\mathcal{N}(0, \frac{1}{m})$. In addition, we
generate a $K$-sparse vector $\mathbf{x}$ whose support is chosen at
random. We consider two types of sparse signals; Gaussian signals
and pulse amplitude modulation (PAM) signals. Each nonzero element
of Gaussian signals is drawn from standard Gaussian and that in PAM
signals is randomly chosen from the set $\{\pm 1, \pm 3\}$.
It should be noted that the gOMP algorithm should satisfy $N \leq K$ and $N \leq m/K$ and thus the $N$ value should not exceed $\sqrt m$ ($N \leq \sqrt m$). In view of this, we choose $N = 3, 6, 9$ in our simulations.
For each recovery algorithm, we perform at least $5,000$ independent
trials and plot the empirical frequency of exact reconstruction.

In Fig. \ref{fig:GaussER}, we provide the recovery performance as a function of the sparsity level $K$.
Clearly, higher critical sparsity implies better empirical reconstruction performance.
The simulation results reveal that the critical sparsity of
the gOMP algorithm is larger than that of the ROMP, OMP, and StOMP algorithms.
Even compared to the LP technique and CoSaMP, the gOMP exhibits slightly better recovery performance.
Fig. \ref{fig:PAMER} provides results for the PAM input signals. We
observe that the overall behavior is similar to the case of Gaussian
signals except that the $\ell_1$-minimization is better than the gOMP.
Overall, we observe that the gOMP algorithm is competitive for both Gaussian and PAM input scenarios.

\begin{figure}[t]
\begin{center}
\includegraphics[width=95mm]{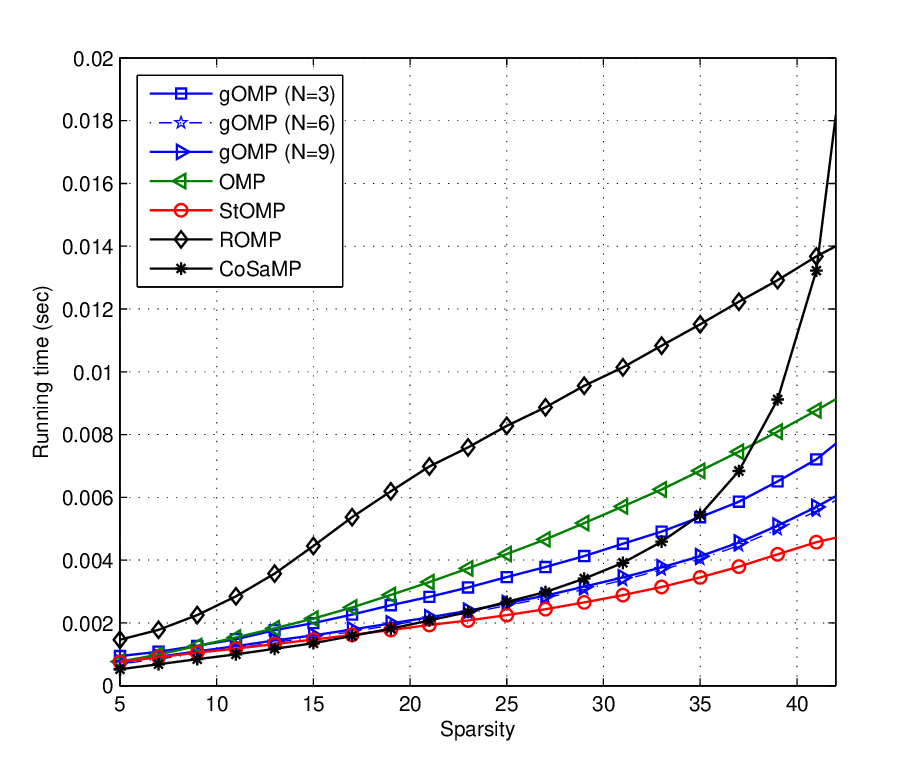}
 \caption{Running time as a function of sparsity $K$. Note that the
running time of the $\ell_1$-minimization is not in the figure since
the time is more than order of magnitude higher than the time of
other algorithms.} \label{fig:run_time}
\end{center}
\end{figure}

In Fig. \ref{fig:run_time}, the running time (average of Gaussian
and PAM signals) for recovery algorithms is provided. The running
time is measured using the MATLAB program under quad-core 64-bit
processor and Window $7$ environments.\footnote{Note that we do not
use any option for enabling the multithread operations.} Note that
we do not add the result of LP technique simply because the running
time is more than order of magnitude higher than that of all other
algorithms.
Overall, we observe that the running time of StOMP, CoSaMP, gOMP,
and OMP is more or less similar when the signal vector is sparse
(i.e., when $K$ is small).
However, when the signal vector becomes less sparse (i.e., when $K$ is
large), the running time of the CoSaMP and OMP increases much faster
than that of the gOMP and StOMP. In particular, while the running
time of the OMP, StOMP, and gOMP increases linearly over $K$, that
for the CoSaMP seems to increase quadratically over $K$.
Among algorithms under test, the running time of the StOMP and gOMP ($N = 6, 9$) is smallest.

\section{RIP based Recovery Condition Analysis}
\label{sec:mainresults}
In this section, we analyze the RIP based condition under which the
gOMP can perfectly recover $K$-sparse signals.
First, we analyze the condition ensuring a success at the first
iteration ($k =1$). Success means that at least one correct index is
chosen in the iteration. Next, we study the condition ensuring the
success in the non-initial iteration ($k > 1$).
By combining two conditions, we obtain the sufficient condition of the gOMP algorithm
guaranteeing the perfect recovery of $K$-sparse signals.
The following lemmas are useful in our analysis.

\begin{lemma}[Lemma 3 in \cite{candes2005decoding,dai2009subspace}] \label{lem:mono}
If the sensing matrix satisfies the RIP of both orders $K_1$ and
$K_2$, then ${\delta _{K_1}} \leq  {\delta _{K_2}}$ for any $K_1
\leq K_2$. This property is referred to as the monotonicity of the
isometry constant.
\end{lemma}

\begin{lemma} [Consequences of RIP \cite{needell2009cosamp,candes2005decoding}]
\label{lem:rips} For $I \subset \Omega$, if $\delta_{\left| I
\right|} < 1$ then for any ${\mathbf{u}} \in {\mathbb{R}^{\left| I
\right|}}$,
\begin{eqnarray}
\left( {1 - {\delta _{\left| I \right|}}} \right){\left\|
{\mathbf{u}} \right\|_2}  \leq {\left\| {{\mathbf{\Phi
}}_I'{{\mathbf{\Phi }}_I}{\mathbf{u}}} \right\|_2}
 \leq \left( {1 + {\delta _{\left| I
\right|}}} \right){\left\| {\mathbf{u}} \right\|_2}, \nonumber \\
\frac {1} { {1 + {\delta _{\left| I \right|}}}}  {\left\|
{\mathbf{u}} \right\|_2}  \leq { \| {  \left({\mathbf{\Phi }}_I'
{{\mathbf{\Phi }}_I}\right)^{-1} {\mathbf{u}}}  \|_2}
 \leq \frac {1} {1 - {\delta _{\left| I
\right|}}} {\left\| {\mathbf{u}} \right\|_2}. \nonumber
\end{eqnarray}
\end{lemma}

\begin{lemma}[Lemma 2.1 in \cite{candes2008restricted} and Lemma 1 in \cite{dai2009subspace}]
\label{lem:correlationrip} Let ${I_1}, {I_2}\subset \Omega$ be two
disjoint sets ($I_1 \cap I_2 = \emptyset$). If ${\delta _{|I_1| +
|I_2|} }< 1$, then
\begin{eqnarray}
{\left\| {{\mathbf{\Phi }}_{{I_1}}' {{\mathbf{\Phi }}}{\mathbf{u}}}
\right\|_2} = {\left\| {{\mathbf{\Phi }}_{{I_1}}' {{\mathbf{\Phi
}}_{{I_2}}}{\mathbf{u}_{I_2}}} \right\|_2}  \leq  {\delta _{|I_1| +
|I_2|} }{\left\| {\mathbf{u}} \right\|_2} \nonumber
\end{eqnarray}
holds for any $\mathbf{u}$ supported on $I_2$.
\end{lemma}

\subsection{Condition for Success at the Initial Iteration}
\label{subsec:ggggg1}
As mentioned, if at least one index is correct among $N$ indices
selected, we say that the gOMP makes a success in the iteration. The
following theorem provides a sufficient condition guaranteeing the
success of the gOMP in the first iteration.

\begin{theorem}\label{thm:firstsuccess}
Suppose $\mathbf{x} \in \mathbb{R}^n$ is a $K$-sparse signal ($K \geq N$), then
the gOMP algorithm makes a success in the first iteration if
%
\begin{eqnarray}  \label{eq:firstsuccess}
\delta_{K + N} < \frac{\sqrt{N}}{\sqrt{K}+\sqrt{N}}.
\end{eqnarray}
\end{theorem}
%
%
\begin{proof}
Let $\Lambda^1$ denote the set of $N$ indices chosen in the first
iteration. Then, elements of $\mathbf{\Phi }_{{{\Lambda^1}}}'
\mathbf{y}$ are $N$ significant elements in ${\mathbf{\Phi
}}'\mathbf{y}$ and thus
\begin{eqnarray} \label{eq:try}
{{\left\| \mathbf{\Phi }_{{{\Lambda^1}}}' \mathbf{y} \right\|}_{2}}
&=& \mathop {\max} \limits_{\left| I \right| = N }
\sqrt{\sum\limits_{i\in I}{{{\left| \left\langle {{\varphi
}_{i}},\mathbf{y} \right\rangle  \right|}^{2}}}}
\end{eqnarray}
where $\varphi_i$ denotes the $i$-th column in $\mathbf{\Phi}$.
Further, we have
\begin{eqnarray} \label{eq:try}
\frac{1}{\sqrt{N}}{{\left\| \mathbf{\Phi }_{{{\Lambda^1}}}'
\mathbf{y} \right\|}_{2}} &=& \frac{1}{\sqrt{N}} \mathop {\max}
\limits_{\left| I \right| = N } \sqrt{\sum\limits_{i\in I}{{{\left|
\left\langle {{\varphi }_{i}},\mathbf{y} \right\rangle
\right|}^{2}}}} \\
& = &  \mathop {\max} \limits_{\left| I \right| = N }
\sqrt{\frac{1}{ |I| }\sum\limits_{i\in I}{{{\left| \left\langle
{{\varphi }_{i}},\mathbf{y} \right\rangle
\right|}^{2}}}} \\
\label{eq:zz232} & \geq &  \sqrt{ \frac{1}{ {|T|}} \sum\limits_{i\in
T}{{{\left| \left\langle {{\varphi }_{i}},\mathbf{y} \right\rangle
\right|}^{2}}}}\\
& = & \frac{ 1 }{\sqrt{K}} {{\left\| \mathbf{\Phi }_{T}'\mathbf{y}
\right\|}_{2}}
\end{eqnarray}
where (\ref{eq:zz232}) is from the fact that the average of $N$-best correlation power is larger than or equal to the average of $K$ (true) correlation power.
Using this together with $ \mathbf{y} = {{\mathbf{\Phi}}_{T}}\mathbf{x}_{T} $, we have
\begin{equation}
{{\left\| \mathbf{\Phi }_{{{\Lambda^1}}}' \mathbf{y} \right\|}_{2}}
\geq  \sqrt{\frac{ N }{K}}{{\left\| \mathbf{\Phi
}_{T}'{{\mathbf{\Phi }}_{T}}\mathbf{x}_{T} \right\|}_{2}}
\geq  \sqrt{\frac{ N }{K}}\left( 1-{{\delta }_{K}}
\right){{\left\| \mathbf{x} \right\|}_{2}} \label{eq:try1}
\end{equation}
where the second inequality is from Lemma \ref{lem:rips}.

On the other hand, when no correct index is chosen in the first
iteration (i.e., $\Lambda^1 \cap T = \emptyset$),
\begin{eqnarray}
{{\left\| \mathbf{\Phi }_{{{\Lambda^1 } }}'\mathbf{y}
\right\|}_{2}}={{\left\| \mathbf{\Phi }_{{{\Lambda^1 }
}}'{{\mathbf{\Phi }}_{T}}{{\mathbf{x}}_{T}} \right\|}_{2}}  \leq
{{\delta }_{K +N}}{{\left\| \mathbf{x} \right\|}_{2}},
\end{eqnarray}
where the inequality follows from Lemma \ref{lem:correlationrip}.
This inequality contradicts (\ref{eq:try1}) if
\begin{eqnarray} \label{eq:gggal}
{{\delta }_{K +N}}{{\left\| \mathbf{x} \right\|}_{2}} < \sqrt{\frac{
N }{K}}\left( 1-{{\delta }_{K}} \right){{\left\| \mathbf{x}
\right\|}_{2}}.
\end{eqnarray}
Note that, under (\ref{eq:gggal}), at least one correct index is
chosen in the first iteration.
Since $\delta_K  \leq \delta_{K + N}$ by Lemma \ref{lem:mono},
\eqref{eq:gggal} holds true when
\begin{eqnarray}
{{\delta }_{K + N}}{{\left\| \mathbf{x} \right\|}_{2}} <
\sqrt{\frac{ N }{K}}\left( 1-{{\delta }_{K + N}} \right){{\left\|
\mathbf{x} \right\|}_{2}}.
\end{eqnarray}
Equivalently,
\begin{eqnarray}
\delta_{K + N} < \frac{\sqrt{N}}{\sqrt{K}+\sqrt{N}}.
\end{eqnarray}
In summary, if $\delta_{K + N} <\frac{\sqrt{N}}{\sqrt{K}+\sqrt{N}}$, then $\Lambda^1$ contains at
least one element of $T$ in the first iteration of the gOMP.
\end{proof}

\subsection{Condition for Success in Non-initial Iterations}
\label{subsec:ggggg2}
In this subsection, we investigate the condition guaranteeing the
success of the gOMP in non-initial iterations.

\vspace{10pt}
\begin{theorem}\label{thm:atleast1}
Suppose $N \leq \min \{K , \frac{m}{K}\}$ and the gOMP has performed $k$ iterations ($1 \leq k < K$) successfully. Then under the condition
\begin{eqnarray}
{\delta _{NK}} < \frac{{\sqrt N }} {{\sqrt K  + 3 \sqrt N }},
\label{eq:good}
\end{eqnarray}
the gOMP will make a success at the $(k + 1)$-th condition.
\end{theorem}

As mentioned, newly selected $N$ indices are not overlapping with
previously selected ones and hence $| \Lambda ^k | = kN$. Also,
under the hypothesis that the gOMP has performed $k$ iterations
successfully, ${\Lambda ^k}$ contains at least $k$ correct indices.
In other words, the number of correct indices $l$ in $\Lambda^k$ becomes
$$l = | T \cap \Lambda^k| \geq k.$$
Note that we only consider the case where $\Lambda^k$ does not
include all correct indices ($ l < K $) since otherwise the
reconstruction task is already finished.\footnote{When all the
correct indices are chosen ($T \subseteq \Lambda^k$) then the
residual $\mathbf{r}^k = \mathbf{0}$ and hence the gOMP algorithm is
finished already.}
Hence, we can safely assume that the set containing the rest of the correct indices is nonempty ($T - {\Lambda ^k} \neq
\emptyset$).

\begin{figure}[t]
\begin{center}
\includegraphics[width=70mm]{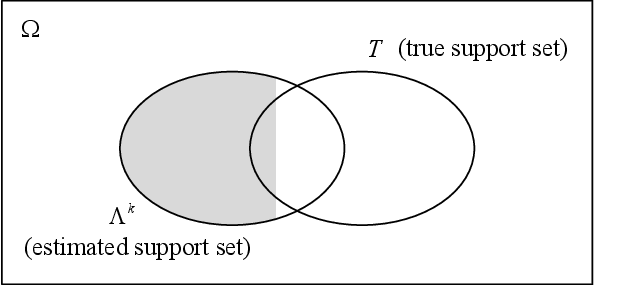}
 \caption{Set diagram of $\Omega$, $T$, and $\Lambda^k$.}
\label{fig:set}
\end{center}
\end{figure}

Key ingredients in our proof are 1) the upper bound $\alpha_N$ of
the $N$-th largest correlation in magnitude between $\mathbf{r}^{k}$ and
columns indexed by $F = \Omega \backslash (\Lambda^k \cup T)$ (i.e.,
the set of remaining incorrect indices) and 2) the lower bound
$\beta_1$ of the largest correlation in magnitude between
${{\mathbf{r}}^k}$ and columns whose indices belong to $T - {\Lambda
^k}$ (i.e., the set of remaining correct indices).
If $\beta_1$ is larger than $\alpha_N$, then $\beta_1$ is contained
in the top $N$ among all values of $| \langle
{{\varphi_j},{{\mathbf{r}}^k}} \rangle |$ and hence at least one
correct index is chosen in the $(k+1)$-th iteration.

The following two lemmas provide the upper bound of $\alpha_N$ and
the lower bound of $\beta_1$, respectively.

%
%
\begin{lemma} \label{lem:upperbound}
Let $  \alpha_i  = | \langle {{\varphi_{\phi(i)}},{{\mathbf{r}}^k}}
\rangle| $ where $\phi(i) = \arg \mathop {\max }\limits_{j:j \in  F
\backslash \left\{ {\phi(1), \cdots, \phi(i - 1) } \right\}} \left|
{\left\langle {{\varphi _j},{{\mathbf{r}}^k}} \right\rangle }
\right| $ so that $\alpha_i$ are ordered in magnitude ($\alpha_1
\geq \alpha_2 \geq \cdots$). Then, in the $(k + 1)$-th iteration in
the gOMP algorithm, $\alpha_N$, the $N$-th
largest correlation in magnitude between ${{\mathbf{r}}^k}$ and $\{
\varphi_i \}_{i \in F}$, satisfies
\begin{equation}
\alpha_N \leq \left( {{\delta _{N + K - l}} + \frac{{{\delta _{N +
Nk}}{\delta _{Nk + K - l}}}} {{1 - {\delta _{Nk}}}}}
\right)\frac{{{{\left\| {{{\mathbf{x}}_{T - {\Lambda ^k}}}}
\right\|}_2}}} {{\sqrt N }}.
\end{equation}
\end{lemma}

\begin{proof}
See Appendix \ref{app:upperbound}.
\end{proof}
\vspace{7pt}
\begin{lemma} \label{lem:lowerbound}
Let $\beta_i  = | \langle \varphi_{\phi(i)} , {\mathbf{r}}^k
\rangle| $ where $\phi(i) = \arg \mathop {\max }\limits_{j:j \in (T
- {\Lambda ^k}) \backslash \left\{{ \phi(1), \cdots, \phi (i - 1)}
\right\}} \left| {\left\langle {{\varphi _j},{{\mathbf{r}}^k}}
\right\rangle } \right| $ so that $\beta_i$ are ordered in magnitude
($\beta_1 \geq \beta_2 \geq \cdots$).  Then in the $(k + 1)$-th
iteration in the gOMP algorithm, $\beta_1$, the largest correlation
in magnitude between ${{\mathbf{r}}^k}$ and $\{\varphi_i \}_{i \in T
- {\Lambda ^k}}$, satisfies
\begin{eqnarray}
\beta_1 &\geq& \left( {1 - {\delta _{K - l}} - \frac{{\sqrt {1 + {\delta _{K - l}}} \sqrt {1 + {\delta _{Nk}}} {\delta _{Nk + K - l}}}}
{{1 - {\delta _{Nk}}}}} \right) \nonumber   \\
        & & \times \frac{{{{\left\| {{{\mathbf{x}}_{T - {\Lambda ^k}}}} \right\|}_2}}}
{{\sqrt {K - l} }}.
\end{eqnarray}
\end{lemma}

\begin{proof}
See Appendix \ref{app:lowerbound}.
\end{proof}

\vspace{10pt}

We now have all ingredients to prove Theorem \ref{thm:atleast1}.

{\bf Proof of Theorem \ref{thm:atleast1}}

\begin{proof}
A sufficient condition under which at least one correct index is
selected at the $(k + 1)$-th step can be described as
\begin{eqnarray}
\alpha_N < \beta_1. \label{eq:goodgood}
\end{eqnarray}

Noting that $1 \leq k \leq l< K$ and $1 < N \leq K$ and also using
the monotonicity of the restricted isometry constant in Lemma
\ref{lem:mono}, we have
\begin{eqnarray}
K - l < NK \nonumber &\rightarrow& \delta _{K - l} < \delta _{NK}, \\
Nk + K - l < NK &\rightarrow&  \delta _{Nk + K - l} < \delta
_{NK}, \nonumber \\
Nk < NK &\rightarrow& \delta _{Nk} < \delta _{NK}, \nonumber  \\
N + Nk \leq NK &\rightarrow&  \delta _{N + Nk} \leq \delta _{NK}.
\label{eq:monoto}
\end{eqnarray}
%
%
%
From Lemma \ref{lem:upperbound} and (\ref{eq:monoto}), we have
\begin{eqnarray} \label{eq:36}
\alpha_N \!\!\!\! &\leq& \!\!\!\! \lefteqn{ \left( {{\delta _{N + K - l}} + \frac{{{\delta _{N +
Nk}}{\delta _{Nk + K - l}}}} {{1 - {\delta _{Nk}}}}} \right) } \nonumber \\
\!\!\! & & \times \frac{{\left\| {{\mathbf{x}}_{T-{\Lambda^{k}}}}
\right\|}_{2}} {{\sqrt N }}  \\
&\leq& \!\!\!\! \left( {{\delta _{NK}} +
\frac{{{\delta _{NK}^2} }} {{1 - {\delta _{NK}}}}}
\right)\frac{{\left\| {{\mathbf{x}}_{T-{\Lambda^{k}}}}
\right\|}_{2}} {{\sqrt N }} \\
&=& \label{eq:36} \!\!\!\! {\frac{{{\delta _{NK} } }} {{1 - {\delta
_{NK}}}}} \frac{{\left\| {{\mathbf{x}}_{T-{\Lambda^{k}}}}
\right\|}_{2}} {{\sqrt N }}.
\end{eqnarray}
%
Also, from Lemma \ref{lem:lowerbound} and (\ref{eq:monoto}), we have
%
\begin{eqnarray}
\beta_1 &\geq& \lefteqn{ \left( {1 - {\delta _{K - l}} - \frac{{\sqrt {1 \!\!+\!\!
{\delta _{K - l}}} \sqrt {1 \!\!+\!\! {\delta _{Nk}}} {\delta _{Nk + K -
l}}}} {{1 - {\delta _{Nk}}}}} \!\! \right) } \nonumber \\
& & \times \frac{{{{\left\|{{{\mathbf{x}}_{T - {\Lambda ^k}}}}\! \right\|}_2}}} {{\sqrt {K - l}}}  \\
&\geq& \left( {1 - {\delta _{NK}} - \frac{{\sqrt {1 + {\delta
_{NK}}} \sqrt {1 + {\delta _{NK}}} {\delta _{NK}}}} {{1 - {\delta
_{NK}}}}} \right) \nonumber \\
& & \times \frac{{{{\left\| {{{\mathbf{x}}_{T - {\Lambda
^k}}}} \right\|}_2}}}
{{\sqrt {K - l} }}  \\
&=& \label{eq:37} \frac{{1 - 3{\delta _{NK}}}} {{1 - {\delta
_{NK}}}}\frac{{{{\left\| {{{\mathbf{x}}_{T - {\Lambda ^k}}}}
\right\|}_2}}} {{\sqrt {K - l} }}.
\end{eqnarray}
%
%
Using \eqref{eq:36} and \eqref{eq:37}, we obtain the sufficient
condition of (\ref{eq:goodgood}) as
\begin{eqnarray} \label{eq:sufficientommp4}
\frac{{1 - 3{\delta _{NK}}}} {{1 - {\delta _{NK}}}}\frac{{{{\left\|
{{{\mathbf{x}}_{T - {\Lambda ^k}}}} \right\|}_2}}} {{\sqrt {K - l}
}} >  {\frac{{{\delta _{NK} } }} {{1 - {\delta _{NK}}}}}
\frac{{\left\| {{\mathbf{x}}_{T-{\Lambda^{k}}}} \right\|}_{2}}
{{\sqrt N }}.
\end{eqnarray}
%
After some manipulations, we have
\begin{eqnarray}
{\delta _{NK}} < \frac{{\sqrt N }} {{\sqrt{K - l}  +  3 \sqrt N }}.
\label{eq:sufficientommp3}
\end{eqnarray}
Since $\sqrt{K - l} < \sqrt{K}$, (\ref{eq:sufficientommp3}) holds if
\begin{eqnarray}
{\delta _{NK}} < \frac{{\sqrt N }} {{\sqrt {K}  + 3 \sqrt N }},
\end{eqnarray}
which completes the proof.
\end{proof}

%
%
%
%
\subsection{Overall Sufficient Condition} \label{subsec:ggggg3}
Thus far, we investigated conditions guaranteeing the success of the
gOMP algorithm in the initial iteration ($k = 1$) and non-initial iterations ($k > 1$). We now
combine these results to describe the sufficient condition of the gOMP algorithm ensuring the
perfect recovery of $K$-sparse signals.

Recall from Theorem \ref{thm:firstsuccess} that the gOMP makes a
success in the first iteration if
\begin{eqnarray}
\delta_{K + N} < \frac{\sqrt{N}}{\sqrt{K}+\sqrt{N}}. \label{eq:A}
\end{eqnarray}
Also, recall from Theorem \ref{thm:atleast1} that if the previous
$k$ iterations were successful,
then the gOMP will be successful for the $(k + 1)$-th iteration if
\begin{eqnarray}
\delta_{ NK } < \frac{\sqrt{N}}{\sqrt{K} + 3 \sqrt{N}}.  \label{eq:B}
\end{eqnarray}
In essence, the overall sufficient condition is determined by the stricter condition between (\ref{eq:A}) and (\ref{eq:B}).

%
%
%
\begin{theorem}[Sufficient condition of gOMP]\label{thm:summary_gOMPgood}
Let $N \leq \min \{K , \frac{m}{K}\}$, then the gOMP algorithm perfectly recovers any $K$-sparse vector $\mathbf{x}$ from $\mathbf{y} = \mathbf{ \Phi x}$ via at most $K$ iterations if the sensing matrix $\mathbf{\Phi }$ satisfies the RIP with isometry constant
\begin{eqnarray}
\label{eq:allsuccess}
&\delta_{ NK } < \frac{\sqrt{N}}{\sqrt{K} + 3 \sqrt{N}} \hspace{1cm} \mbox{for}~ K > 1, \\
&\delta_2      < \frac{1}{2} \hspace{2.5cm} \mbox{for}~ K = 1.
\end{eqnarray}
\end{theorem}

\begin{proof}
In order to prove the theorem, the following three cases need to be
considered.
\begin{itemize}
%
%
\item Case 1 [$N > 1, K > 1$]: \\
In this case, $NK \geq K + N$ and hence $\delta _{NK} \geq \delta
_{K + N}$ and also $ \frac{{\sqrt N }} {{\sqrt K + \sqrt N }} >
\frac{{\sqrt N }} {{\sqrt K  +  3 \sqrt N }}$. Thus, (\ref{eq:B}) is
stricter than (\ref{eq:A}) and the general condition becomes
\begin{eqnarray} \label{eq:43}
\delta_{ NK } < \frac{\sqrt{N}}{\sqrt{K} + 3 \sqrt{N}}.
\end{eqnarray}

%
%
\item Case 2 [$N = 1, K > 1$]: \\In this case, the
general condition should be the stricter condition between ${\delta
_{K}} < \frac{1} {{\sqrt K + 3}}$ and $\delta_{K + 1} < \frac{1}
{\sqrt K + 1}$.
Unfortunately, since $\delta_K \leq \delta_{K + 1}$ and
$\frac{1}{\sqrt{K} + 3} \leq \frac{1}{\sqrt{K} + 1}$, one cannot
compare two conditions directly. As an indirect way, we borrow a
sufficient condition guaranteeing the perfect recovery of the gOMP
for $N = 1$ as
\begin{eqnarray} \label{eq:goodbound}
{{\delta }_{K}} < \frac{\sqrt{K-1}}{\sqrt{K-1}+K}.
\end{eqnarray}
Readers are referred to \cite{wang2012Near} for the proof of
(\ref{eq:goodbound}).
Since $\frac{1}{{\sqrt K + 3}}  < \frac{\sqrt{K-1}}{\sqrt{K-1}+K}$
for $K > 1$, the sufficient condition for Case 2 becomes
\begin{eqnarray} \label{eq:44}
\delta_{K} < \frac{1}{\sqrt K + 3}.
\end{eqnarray}
It is interesting to note that \eqref{eq:44} can be nicely combined with the result of Case 1 since applying $N = 1$ in \eqref{eq:43} will result in \eqref{eq:44}.

%
%
\item Case 3 [$K = 1$]: \\Since $\mathbf{x}$ has a
single nonzero element ($K = 1$), $\mathbf{x}$ should be recovered in the
first iteration. Let $u$ be the index of nonzero element, then the
exact recovery of $\mathbf{x}$ is ensured regardless of $N$
if
\begin{eqnarray}
\label{eq:best} \left| \left\langle {{\varphi}_{u}},\mathbf{y}
\right\rangle  \right| = \max_i \left| \left\langle
{{\varphi}_{i}},\mathbf{y} \right\rangle  \right|.
\end{eqnarray}
The condition ensuring (\ref{eq:best}) is obtained by applying $N =
K = 1$ for Theorem \ref{thm:firstsuccess} and is given by
$\delta_{2} < \frac{1}{2}.$
\end{itemize}
\end{proof}

\begin{remark} [$\delta_{2K}$ based recovery condition]
We can express our condition with a small order of isometry constant.
By virtue of \cite[Corollary 3.4]{needell2009cosamp} ($\delta_{cK} \leq c \delta_{2K}$ for positive integer $c$), the proposed bound holds whenever $\delta_{2K} < \frac{1}
{\sqrt{NK} + 3 N}$.

\end{remark}

\vspace{3mm}

\begin{remark} [Comparison with previous work]\label{rem:comparison}
It is worth mentioning that there have been previous efforts to
investigate the sufficient condition for this algorithm.
In particular, the condition $\delta_{NK} < \frac{\sqrt{N}} {(2 + \sqrt
2) \sqrt K}$ was established in \cite[Theorem
2.1]{liu2012orthogonal}.
The proposed bound $\delta_{NK} < \frac{\sqrt{N}} {\sqrt{K} + 3 \sqrt{N}}$ holds the advantage over this bound if $N < \frac{(3 + 2\sqrt{2})K} {9} \approx 0.65 K$.
Since $N$ is typically much smaller than $K$, the proposed bound offers better recovery condition in many practical scenarios.\footnote{In fact, $N$ should not be large since the
inequality $N \leq \frac{m}{K}$ must be guaranteed.}
\end{remark}

\vspace{3mm}

\begin{remark} [Measurement size of sensing matrix] \label{rem:size}
It is well known that an $m \times n$ random sensing matrix whose entries are i.i.d.
with Gaussian distribution $N(0, \frac{1}{m})$ obeys the RIP ($\delta_K < \varepsilon$) with overwhelming probability if the dimension of the measurements satisfies \cite{baraniuk2008simple}
\begin{eqnarray} \label{eq:bara}
m = O \left(\frac{ K \log\frac{n}{K}}{\varepsilon^2}\right).
\end{eqnarray}
In \cite{davenport2010analysis}, it is shown that the OMP requires
$m = O \left( K^2 \log ({n/K})\right)$ random measurements for
reconstructing $K$-sparse signal.
Plugging \eqref{eq:allsuccess} into \eqref{eq:bara}, we also get the
same result.
\end{remark}

\subsection{Sufficient Condition of OMP}
\label{sec:mainresult_OMP}
In this subsection, we put our focus on the OMP algorithm which is the
special case of the gOMP algorithm for $N = 1$.
For sure, one can immediately obtain the condition of the OMP $\delta_{K} < \frac{1}{\sqrt{K} + 3}$ by applying $N = 1$ to Theorem \ref{thm:summary_gOMPgood}.
Our result is an improved version of this and based on the fact that
the non-initial step of the OMP process is the same as the initial
step since the residual is considered as a new measurement
preserving the sparsity $K$ of an input vector $\mathbf{x}$
\cite{wang2012Near,wang2012Recovery}.
In this regard, a condition guaranteeing to select a correct index in the first iteration can be readily extended to the general condition without incurring any loss.

\begin{corollary}[Direct consequence of Theorem \ref{thm:firstsuccess}] \label{cor:firstsuccess}
Suppose $\mathbf{x} \in \mathbb{R}^{n}$ is $K$-sparse, then the OMP
algorithm recovers an index in $T$ from $\mathbf{y} = \mathbf{\Phi
x} \in \mathbb{R}^{m}$ in the first iteration if
$\delta_{K + 1} < \frac{1}{\sqrt{K} + 1}$.
\end{corollary}

We now state that the first iteration condition is extended to any
iteration of the OMP algorithm.
\begin{lemma}[Wang and Shim\cite{wang2012Recovery}] \label{lem:good1}
Suppose that the first $k$ iterations ($1 \leq k \leq K - 1$) of the
OMP algorithm are successful (i.e., $\Lambda^k \subset T$), then the
$(k + 1)$-th iteration is also successful (i.e., $t^{k + 1} \in T$)
under $\delta_{K+1} < \frac{1}{\sqrt{K} + 1 }$.
\end{lemma}

Combining Corollary \ref{cor:firstsuccess} and Lemma \ref{lem:good1}, and also noting that indices in $\Lambda^k$ are not selected again in the succeeding iterations (since the index chosen in the $(k + 1)$-th step belongs to $T - \Lambda^k$), one can conclude that $\Lambda^K = T$ and the OMP algorithm recovers original signal $\mathbf{x}$ in
exactly $K$ iterations under $\delta_{K + 1} < \frac{1}{\sqrt{K} +
1}$.

The following theorem formally describes the sufficient condition of
the OMP algorithm.
\begin{theorem}[Wang and Shim\cite{wang2012Recovery}] \label{thm:totalOMP}
Suppose $ \mathbf{x} $ is $K$-sparse vector, then the OMP algorithm
recovers $\mathbf{x}$ from $\mathbf{y} = \mathbf{\Phi x}$ under
\begin{eqnarray}
\delta_{K + 1} < \frac{1}{\sqrt{K} + 1}. \label{eq:nearoptimalbound}
\end{eqnarray}
\end{theorem}
\begin{proof}
Immediate from Corollary \ref{cor:firstsuccess} and Lemma \ref{lem:good1}.
\end{proof}

\vspace{10pt}

\begin{remark} [Comments on bounds of \eqref{eq:43} and \eqref{eq:nearoptimalbound}] \label{rem:whichisbetter}
The bounds of the gOMP in \eqref{eq:43} and the OMP in \eqref{eq:nearoptimalbound} cannot be directly compared since $\delta_{NK} \geq \delta_{K + 1}$ and $\frac{\sqrt N}{\sqrt{K} + 3
\sqrt N} \geq \frac{1}{\sqrt{K} + 1}$.
Nevertheless, the difference in the order might offer possible advantages to the OMP.
It should be noted that the RIP condition, analyzed based on the worse case scenario, is for the perfect recovery and hence offers too conservative bound.
This explains why the most of sparse recovery algorithms perform better than the bound predicts in practice.
It should also be noted that by allowing more iterations (larger
than $K$ iterations) for the OMP, one can improve performance
\cite{liu2012orthogonal,zhang2011sparse,foucart2011stability} and
achieve performance comparable to the gOMP. However, this may incur
large delay and higher computational cost.
\end{remark}

%
%
\section{Reconstruction of Sparse Signals from Noisy Measurements}
%
%
In this section, we consider the reconstruction performance of the
gOMP algorithm in the presence of noise.
Since the measurement is expressed as ${\mathbf{y}} = {\mathbf{\Phi x}} + {\mathbf{v}}$ in this scenario, perfect reconstruction of $\mathbf{x}$ cannot be guaranteed and hence we need to use the upper bound of $\ell_2$-norm distortion ${\left\| {{\mathbf{x}} - {\mathbf{\hat x}}} \right\|_2}$ as a performance measure.

Recall that the termination condition of the gOMP algorithm is either ${\| {{{\mathbf{r}}^s}} \|_2} < \epsilon $ or $k \geq  \min \left\{ {K,\frac{m} {N}} \right\}$.
Note that, since $\min \left\{ {K,\frac{m}{N}} \right\} = K$ under the condition that ${\mathbf{\Phi }}$ satisfies the RIP of order $NK$,\footnote{If ${\mathbf{\Phi }}$ satisfies the RIP of order $NK$, (i.e., $\delta_{NK} \in (0,1)$), then $0 < 1 - \delta_{NK} \leq \lambda_{\min} \left(\mathbf{\Phi}_D'\mathbf{\Phi}_D\right)$ for all index set $D$ with $|D| \leq NK$, which indicates that all eigenvalues of $\mathbf{\Phi}_D$ are positive. Thus, $\mathbf{\Phi}_D$ should be full column rank (i.e., $K \leq  m/N$).}
the stopping rule of the gOMP can be simplified to ${\| {{{\mathbf{r}}^k}} \|_2} < \epsilon $ or ${k} = K$.  In these two scenarios, we investigate the upper bound of ${\left\| {{\mathbf{x}} - {\mathbf{\hat x}}} \right\|_2}$.

The next theorem provides the upper bound of $\ell_2$-norm distortion when the gOMP algorithm is finished by ${\| {{{\mathbf{r}}^k}} \|_2} < \epsilon$.
%
%
\begin{theorem} \label{thm:errorbound1}
%
%
Let ${\mathbf{\Phi }}$ be the sensing matrix satisfying RIP of order $NK$. If ${\left\| {{{\mathbf{r}}^s}} \right\|_2} < \epsilon $ is satisfied after $s$ ($s < K$) iterations, then
\begin{eqnarray}
  {\left\| {{\mathbf{x}} - {\mathbf{\hat x}}} \right\|_2} \leq  \frac{\epsilon }
{{\sqrt {1 - {\delta _{NK}}} }} + \frac{{{{\left\| {\mathbf{v}} \right\|}_2}}}
{{\sqrt {1 - {\delta _{NK}}} }}.
\end{eqnarray}
\end{theorem}

\begin{proof}
See Appendix \ref{app:errorbound1}.
\end{proof}

The next theorem provides the upper bound of ${\left\| {{\mathbf{x}} - {\mathbf{\hat x}}} \right\|_2}$ for the second scenario (i.e., when the gOMP is terminated after $K$ iterations).

%
%
\begin{theorem} \label{thm:errorbound1_two111}
%
%
Let $\mathbf{\Phi}$ be the sensing matrix satisfying the RIP of
order $NK + K$ and ${\delta _{NK}} < \frac{\sqrt N} {3\sqrt N  +
\sqrt K }$.  Suppose the gOMP algorithm is terminated after $K$
iterations, then
\begin{eqnarray}
 {\left\| {{\mathbf{x}} - {\mathbf{\hat x}}} \right\|_2} &\leq&  \frac{{{{\left\| {\mathbf{v}} \right\|}_2}}} {{\sqrt {1 - {\delta _{NK}}} }}, ~~\mbox{if}~~ T \subset \Lambda^K,\\
 {\left\| {{\mathbf{x}} - {\mathbf{\hat x}}} \right\|_2} &\leq&  C_K {\left\| {\mathbf{v}} \right\|_2}, ~~~~~\mbox{if}~~ T \not\subset \Lambda^K,
\end{eqnarray}
where
\begin{eqnarray}
C_K & = & \frac{\left( 1-{{\delta }_{NK}} \right)\left( 1+{{\delta
}_{K}}+\sqrt{\frac{K}{N} \left( 1+{{\delta }_{N}} \right)\left(
1+{{\delta }_{K}} \right)} \right)}{\left( 1-3{{\delta
}_{NK}}-\sqrt{\frac{K}{N} }{{\delta }_{NK}} \right)\sqrt{1-{{\delta }_{NK+K}}}} \nonumber    \\
& & + \frac{2\left( 1-3{{\delta}_{NK}}-\sqrt{\frac{K}{N} }{{\delta
}_{NK}} \right)} {\left( 1-3{{\delta }_{NK}}-\sqrt{\frac{K}{N}
}{{\delta }_{NK}} \right)\sqrt{1-{{\delta }_{NK+K}}}}.  \nonumber
\end{eqnarray}
Since $C_K > \frac{1}{\sqrt{1 - \delta_{NK}}}$, one can get the simple upper bound as
$${\left\| {{\mathbf{x}} - {\mathbf{\hat x}}} \right\|_2} \leq  C_K {\left\| {\mathbf{v}} \right\|_2}.$$
\end{theorem}

\begin{remark}[Comments about $C_K$] \label{rem:ck}
Note that, by the hypothesis of the theorem, $0 < 1 - 3 \delta_{NK}
- \sqrt{\frac{K}{N} }{{\delta }_{NK}} < 1$, so $C_k$ for large $K$
is approximately
\begin{equation}
\label{eq:approx_ck}
C_K \approx \frac{ \left( 1-{{\delta }_{NK}} \right) \sqrt{
\left( 1+{{\delta }_{N}} \right)\left( 1+{{\delta }_{K}} \right)}
}{\left( 1-3{{\delta }_{NK}}-\sqrt{\frac{K}{N} }{{\delta }_{NK}} \right)\sqrt{1-{{\delta }_{NK+K}}}}
\sqrt{\frac{K}{N}}.
\end{equation}
This says that the $\ell_2$-norm distortion is essentially upper bounded
by the product of noise power and $c \sqrt{\frac{K}{N}}$ ($c$ is a constant).
It is clear from \eqref{eq:approx_ck} that $C_K$ decreases as $N$ increases. Hence,
by increasing $N$ (i.e., allowing more indices to be selected per
step), we may obtain a better (smaller) distortion bound. However,
since $NK \leq m$ needs to be satisfied, this bound is guaranteed only
for very sparse signal vectors.
\end{remark}

\vspace{10pt}

Before providing the proof of Theorem \ref{thm:errorbound1_two111}, we
analyze a sufficient condition of the gOMP to make success at the
$(k + 1)$-th iteration when the former $k$ iterations are
successful. In our analysis, we reuse the notation $\alpha_i$ and
$\beta_i$ of Lemma \ref{lem:upperbound} and \ref{lem:lowerbound}.
The following lemma provides an upper bound of $\alpha_N$ and a lower
bound of $\beta_1$.
%
%
\begin{lemma}  \label{lem:errorbound1_two}
$\alpha_N$ and $\beta_1$ satisfy
\begin{eqnarray}
{\alpha _N} &\leq&  \left( {{\delta _{N + K - l}} + \frac{{{\delta _{N + Nk}}{\delta _{Nk + K - l}}}}
{{1 - {\delta _{Nk}}}}} \right)\frac{{{{\left\| {{{\mathbf{x}}_{T - {\Lambda ^k}}}} \right\|}_2}}}
{{\sqrt N }} \nonumber  \\
& & + \frac{{\sqrt {1 + {\delta _N}} {{\left\| {\mathbf{v}} \right\|}_2}}}
{{\sqrt N }}
\end{eqnarray}
and
\begin{eqnarray}
{\beta _1} &\geq&  \left( {1 - {\delta _{K - l}} - \frac{{\sqrt {1 + {\delta _{K - l}}} \sqrt {1 + {\delta _{Nk}}} {\delta _{Nk + K - l}}}}
{{1 - {\delta _{Nk}}}}} \right) \nonumber   \\
& & \times \frac{{{{\left\| {{{\mathbf{x}}_{T - {\Lambda ^k}}}} \right\|}_2}}}
{{\sqrt {K - l} }} - \frac{{\sqrt {1 + {\delta _{K - l}}} {{\left\| {\mathbf{v}} \right\|}_2}}}
{{\sqrt {K - l} }}.
\end{eqnarray}
\end{lemma}

\begin{proof}
See Appendix \ref{app:errorbound1_two}.
\end{proof}

As mentioned, the gOMP will select at least one correct index from $T$ at the $(k + 1)$-th iteration provided that $\alpha_N < \beta_1$.

%
%
\begin{lemma} \label{lem:errorbound1_two1}
Suppose the gOMP has performed $k$ iterations ($1 \leq k < K$)
successfully. Then under the condition
\begin{equation} \label{eq:ggaalssl}
{{\left\| {{\mathbf{x}}_{T-{{\Lambda }^{k}}}} \right\|}_{2}}>\frac{\left( \sqrt{1+{{\delta }_{K}}}+\sqrt{\frac{K}{N} +{\frac{K}{N} {\delta }_{N}} } \right)\left( 1-{{\delta }_{NK}} \right)}{1-3{{\delta }_{NK}}-\sqrt{\frac{K}{N} }{{\delta }_{NK}}}{{\left\| \mathbf{v} \right\|}_{2}},
\end{equation}
the gOMP makes a success at the $(k + 1)$-th condition.
\end{lemma}

\begin{proof}
See Appendix \ref{app:errorbound1_two1}.
\end{proof}

We are now ready to prove Theorem \ref{thm:errorbound1_two111}.

\begin{proof}
We first consider the scenario where $\Lambda ^K$ contains all correct indices (i.e.,  $T \subset {\Lambda ^K}$). In this case,
\begin{eqnarray}
 \!\!\!\!\!\! {\left\| {{\mathbf{x}} - {\mathbf{\hat x}}} \right\|_2} \!\!\!\!\! & \leq & \!\!\!\!\! \frac{1}
{{\sqrt {1 - {\delta _{\left| {{\Lambda ^K} \cup T} \right|}}} }}{\left\| {{\mathbf{\Phi }}\left( {{\mathbf{x}} - {\mathbf{\hat x}}} \right)} \right\|_2} \\
&=& \!\!\!\!\! \frac{1}
{{\sqrt {1 - {\delta _{\left| {{\Lambda ^K}} \right|}}} }}{\left\| {{\mathbf{\Phi }}\left( {{\mathbf{x}} - {\mathbf{\hat x}}} \right)} \right\|_2} \\
&=& \!\!\!\!\! \frac{1}
{{\sqrt {1 - {\delta _{NK}}} }}{\left\| {{\mathbf{\Phi x}} - {{\mathbf{\Phi }}_{{\Lambda ^K}}}{\mathbf{\Phi }}_{{\Lambda ^K}}^\dag {\mathbf{y}}} \right\|_2}   \\
&=& \!\!\!\!\! \frac{\left\| {{\mathbf{\Phi x}} - {{\mathbf{\Phi }}_{{\Lambda ^K}}}{\mathbf{\Phi }}_{{\Lambda ^K}}^\dag \left( {{\mathbf{\Phi x}} + {\mathbf{v}}} \right)} \right\|_2}
{{\sqrt {1 - {\delta _{NK}}} }} \\
&=& \!\!\!\!\! \frac{\| {{\mathbf{\Phi x}} \!-\! {{\mathbf{\Phi }}_{{\Lambda ^K}}}{\mathbf{\Phi }}_{{\Lambda ^K}}^\dag {\mathbf{\Phi x}} \!-\! {{\mathbf{\Phi }}_{{\Lambda ^K}}}{\mathbf{\Phi }}_{{\Lambda ^K}}^\dag {\mathbf{v}}} \|_2}
{{\sqrt {1 - {\delta _{NK}}} }}   \\
\label{eq:llssls3}&=& \!\!\!\!\! \frac{\left\| {{{\mathbf{P}}_{{\Lambda ^K}}}{\mathbf{v}}} \right\|_2}
{{\sqrt {1 - {\delta _{NK}}} }} \\
 & \leq & \!\!\!\!\! \frac{{{{\left\| {\mathbf{v}} \right\|}_2}}}
{{\sqrt {1 - {\delta _{NK}}} }},
\end{eqnarray}
where \eqref{eq:llssls3} follows from $
{\mathbf{\Phi x}} - {{\mathbf{\Phi }}_{{\Lambda ^K}}}{\mathbf{\Phi }}_{{\Lambda ^K}}^\dag {\mathbf{\Phi x}} = {\mathbf{\Phi x}} - {{\mathbf{\Phi }}_{{\Lambda ^K}}}{\mathbf{\Phi }}_{{\Lambda ^K}}^\dag {{\mathbf{\Phi }}_{{\Lambda ^K}}}{{\mathbf{x}}_{{\Lambda ^K}}} = {\mathbf{0}}$
for $T \subset {\Lambda ^K}$.

Now we turn to the next scenario where $\Lambda ^K$ does not contain all the correct indices (i.e., $T \not\subset \Lambda^K $). Since the algorithm has performed $K$
iterations yet failed to find all correct indices, it is clear
that the gOMP algorithm does not make a success for some iteration (say this occurs at $(p+1)$-th iteration).
Then, by the contraposition of Lemma \ref{lem:errorbound1_two1},
\begin{equation} \label{eq:ggaall}
{\left\| {{{\mathbf{x}}_{T - {\Lambda ^p}}}} \right\|_2} \leq \frac{\left( \sqrt{1+{{\delta }_{K}}}+\sqrt{\frac{K}{N}+\frac{K}{N}{{\delta }_{N}}} \right)\left( 1-{{\delta }_{NK}} \right)}{1-3{{\delta }_{NK}}-\sqrt{\frac{K}{N}}{{\delta }_{NK}}}{{\left\| \mathbf{v} \right\|}_{2}}.
\end{equation}
Since ${{\mathbf{x}} - {\mathbf{\hat x}}}$ is at most $(NK +K)$-sparse,
\begin{eqnarray}
  {\left\| {{\mathbf{x}} - {\mathbf{\hat x}}} \right\|_2} &\leq&  \frac{1}
{{\sqrt {1 - {\delta _{\left| {{\Lambda ^K} \cup T} \right|}}} }}{\left\| {{\mathbf{\Phi }}\left( {{\mathbf{x}} - {\mathbf{\hat x}}} \right)} \right\|_2}   \\
&\leq&  \frac{\left\| {{\mathbf{\Phi x}} - {{\mathbf{\Phi }}_{{\Lambda ^K}}}{\mathbf{\Phi }}_{{\Lambda ^K}}^\dag {\mathbf{y}}} \right\|_2}
{{\sqrt {1 - {\delta _{NK + K}}} }}   \\
&=&  \frac{\left\| {{\mathbf{y - v}} - {{\mathbf{\Phi }}_{{\Lambda ^K}}}{\mathbf{\Phi }}_{{\Lambda ^K}}^\dag {\mathbf{y}}} \right\|_2}
{{\sqrt {1 - {\delta _{NK + K}}} }}   \\
&=& \frac{\left\| {{{\mathbf{r}}^K} - {\mathbf{v}}} \right\|_2}
{{\sqrt {1 - {\delta _{NK + K}}} }}   \\
&\leq&  \frac{{{{\left\| {{{\mathbf{r}}^K}} \right\|}_2} + {{\left\|
{\mathbf{v}} \right\|}_2}}} {{\sqrt {1 - {\delta _{NK + K}}} }}.
\end{eqnarray}
Also, ${\left\| {{{\mathbf{r}}^K}} \right\|_2} \leq  {\left\|
{{{\mathbf{r}}^p}} \right\|_2}$,\footnote{Due to the orthogonal projection of the gOMP, the magnitude of the residual decreases as iterations go on (${\left\| {{{\mathbf{r}}^i}} \right\|_2} \leq  {\left\| {{{\mathbf{r}}^j}} \right\|_2}$ for $i \geq j$).} and thus
\begin{eqnarray}
 \!\!\! {\left\| {{\mathbf{x}} - {\mathbf{\hat x}}} \right\|_2} \!\!\! &\leq& \!\!\! \frac{1}
{{\sqrt {1 - {\delta _{NK + K}}} }}\left( {{{\left\| {{{\mathbf{r}}^p}} \right\|}_2} + {{\left\| {\mathbf{v}} \right\|}_2}} \right)  \\
&=&  \frac{{{{\left\| {{\mathbf{P}}_{{\Lambda ^p}}^ \bot {\mathbf{y}}} \right\|}_2} + {{\left\| {\mathbf{v}} \right\|}_2}}}
{{\sqrt {1 - {\delta _{NK + K}}} }}     \\
&=&  \frac{{{{\left\| {{\mathbf{P}}_{{\Lambda ^p}}^ \bot {{\mathbf{\Phi }}_{T}}{{\mathbf{x}}_{T}} + {\mathbf{P}}_{{\Lambda ^p}}^ \bot {\mathbf{v}}} \right\|}_2} + {{\left\| {\mathbf{v}} \right\|}_2}}}
{{\sqrt {1 - {\delta _{NK + K}}} }} \\
\label{eq:hhshshhshs}  &=&  \frac{{{{\left\| {{\mathbf{P}}_{{\Lambda ^p}}^ \bot {{\mathbf{\Phi }}_{T - {\Lambda ^p}}}{{\mathbf{x}}_{T - {\Lambda ^p}}} + {\mathbf{P}}_{{\Lambda ^p}}^ \bot {\mathbf{v}}} \right\|}_2}}}
{{\sqrt {1 - {\delta _{NK + K}}} }} \nonumber   \\
& &  +\frac{ \left\| {\mathbf{v}} \right\|_2} {\sqrt {1 - {\delta _{NK + K}}}}  \\
&\leq&  \frac{\left\| {{\mathbf{P}}_{{\Lambda ^p}}^ \bot {{\mathbf{\Phi }}_{T - {\Lambda ^p}}}{{\mathbf{x}}_{T - {\Lambda ^p}}}} \right\|_2}
{\sqrt {1 - {\delta _{NK + K}}}}    \nonumber   \\
& &+\frac{ \left\| {{\mathbf{P}}_{{\Lambda ^p}}^ \bot {\mathbf{v}}} \right\|_2 + \left\| {\mathbf{v}} \right\|_2 }{\sqrt {1 - {\delta _{NK + K}}}}  \\
\label{eq:giays2}&\leq&  \frac{ {{{\left\| {{{\mathbf{\Phi }}_{T -
{\Lambda ^p}}}{{\mathbf{x}}_{T - {\Lambda ^p}}}} \right\|}_2} +
2{{\left\| {\mathbf{v}} \right\|}_2}}} {{\sqrt {1 - {\delta _{NK +
K}}} }},
\end{eqnarray}
where \eqref{eq:hhshshhshs} is because ${\mathbf{P}}_{\Lambda
^p}^\bot$ cancels all the components in
$span(\mathbf{\Phi}_{\Lambda^p})$.
Using the definition of the RIP, we have
\begin{eqnarray}
{\left\| {{{\mathbf{\Phi }}_{T - {\Lambda ^p}}}{{\mathbf{x}}_{T - {\Lambda ^p}}}} \right\|_2}
    &\leq&  \sqrt {1 + {\delta _{\left| {T - {\Lambda ^p}} \right|}}} {\left\| {{{\mathbf{x}}_{T - {\Lambda ^p}}}} \right\|_2}  \nonumber   \\
    &\leq&  \sqrt {1 + {\delta _K}} {\left\| {{{\mathbf{x}}_{T - {\Lambda ^p}}}} \right\|_2},
\end{eqnarray}
and hence 
\begin{equation}
\label{eq:vvv33}
{\left\| {{\mathbf{x}} - {\mathbf{\hat x}}} \right\|_2}
\leq
\frac{{\sqrt {1 + {\delta _K}} {{\left\| {{{\mathbf{x}}_{T - {\Lambda ^p}}}} \right\|}_2} + 2{{\left\| {\mathbf{v}} \right\|}_2}}}
{{\sqrt {1 - {\delta _{NK + K}}} }}.
\end{equation}

Combining \eqref{eq:vvv33} and \eqref{eq:ggaall}, we finally have
\begin{eqnarray}
{\left\| {{\mathbf{x}} - {\mathbf{\hat x}}} \right\|_2} \leq
C_K{\left\| {\mathbf{v}} \right\|_2},
\end{eqnarray}
where
\begin{eqnarray}
C_K & = & \frac{\left( 1-{{\delta }_{NK}} \right)\left( 1+{{\delta
}_{K}}+\sqrt{\frac{K}{N} \left( 1+{{\delta }_{N}} \right)\left(
1+{{\delta }_{K}} \right)} \right)}{\left( 1-3{{\delta
}_{NK}}-\sqrt{\frac{K}{N} }{{\delta }_{NK}} \right)\sqrt{1-{{\delta }_{NK+K}}}} \nonumber    \\
& & + \frac{2\left( 1-3{{\delta}_{NK}}-\sqrt{\frac{K}{N} }{{\delta
}_{NK}} \right)} {\left( 1-3{{\delta }_{NK}}-\sqrt{\frac{K}{N}
}{{\delta }_{NK}} \right)\sqrt{1-{{\delta }_{NK+K}}}}.  \nonumber
\end{eqnarray}
\end{proof}


%
%
\section{Discussions on Complexity}
%
In this section, we discuss the complexity of the gOMP algorithm.  Our analysis shows that the computational complexity is approximately $2s mn + (2N^2 + N)s^2m$ where $s$ is the number of iterations. We show by empirical simulations that the number of iterations $s$ is small so that the proposed gOMP algorithm is effective in running time and computational complexity.

%
%
%
The complexity for each step of the gOMP algorithm is summarized as
follows.

\begin{itemize}
\item \textbf{Identification}: The gOMP performs a
matrix-vector multiplication $\mathbf{\Phi}' \mathbf{r}^{k - 1}$, which requires $(2m - 1)n$ flops ($m$ multiplication and $m - 1$ additions).
Also, $\mathbf{\Phi}' \mathbf{r}^{k - 1}$ needs to be sorted to find $N$ best indices, which requires $nN - N(N + 1)/2$ flops.

\item \textbf{Estimation of $\hat{\mathbf{x}}_{\Lambda^k}$}:
The LS solution $\hat{\mathbf{x}}_{\Lambda^k}$ is obtained in this step.
Using the QR factorization of $\mathbf{\Phi}_{\Lambda^k}$
($\mathbf{\Phi}_{\Lambda^k} = \mathbf{Q R}$),  we have
\begin{equation}
{{\mathbf{\hat x}}_{\Lambda^{k}}} =
\left(\mathbf{\Phi}_{\Lambda^k}'\mathbf{\Phi}_{\Lambda^k}\right)^{-1}
\mathbf{\Phi}_{\Lambda^k}' \mathbf{y} = \left(\mathbf{R}'
\mathbf{R}\right)^{-1} \mathbf{R}' \mathbf{Q}' \mathbf{y}
\end{equation}
and this leads to a cost of $O(k^2 m)$ \cite{bj?rck1996numerical}.
Actually, since the elements of $\Lambda^{k}$ and $\Lambda^{k - 1}$
are largely overlapped, it is possible to recycle the part of the
previous QR factorization of $\mathbf{\Phi}_{\Lambda^{k - 1}}$ and then apply the modified Gram-Schmidt (MGS) algorithm.
In doing so, the LS solution can be solved efficiently (see Appendix
\ref{app:LS}),\footnote{We note that the fast approach of the MGS for
solving LS problem can also be applied to the OMP, ROMP, and StOMP, but with the exception of the CoSaMP. This is because the CoSaMP algorithm computes a completely new LS solution over the distinct subset of $\mathbf{\Phi}$ in each iteration. Nevertheless, other fast
approaches such as Richardson's iteration or conjugate gradient might be applied in the CoSaMP.}
and the associated cost is $4N^2 km + (- 2 N^2 + 5N )m + 2N^3 k^2 +
(-4N^3 + 5 N ^2)k + 3 N^3 - N ^2 - N$ flops.

\item \textbf{Residual update}: For the residual update, the gOMP performs the matrix-vector multiplication $\mathbf{\Phi}_{\Lambda^k} \hat{\mathbf{x}}_{\Lambda^k}$ ($(2Nk - 1)m$ flops) followed by the subtraction ($m$ flops).
\end{itemize}

\begin{figure}[t]
\begin{center}
\includegraphics[width=95mm]{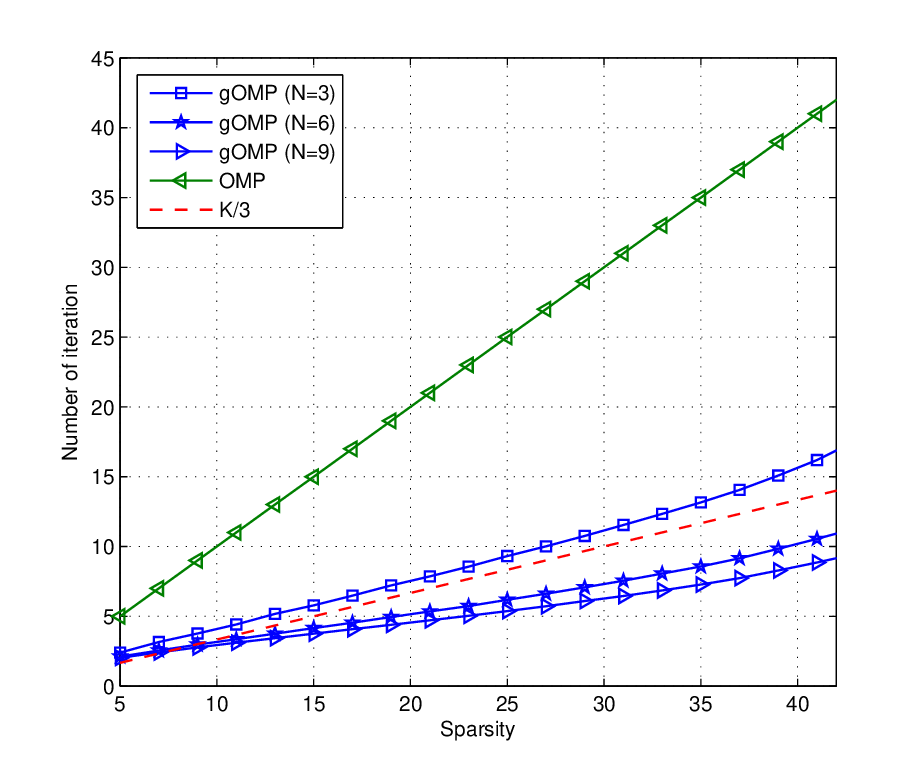}
 \caption{Number of iterations of the OMP and gOMP ($N = 5$) as a
function of sparsity $K$. The red dashed line is the reference curve
indicating $K/3$ iterations.}\label{fig:Niter}
\end{center}
\end{figure}

Table \ref{tab:gOMP_cost} summarizes the complexity of each
operation in the $k$-th iteration of the gOMP. The complexity of the
$k$-th iterations is approximately $2mn + (4N^2 + 2N)km$.
If the gOMP is finished in $s$ iterations, then the complexity of
the gOMP, denoted as $\mathcal{C}_{gOMP}(N,s,m,n)$, becomes
\begin{eqnarray}
\mathcal{C}_{gOMP}(N,s,m,n) & \approx & \sum_{k = 1}^s 2mn + (4N^2 +
2N)km \nonumber \\
    &=& 2s mn + (2N^2 + N)s^2m. \label{eq:costbound} \nonumber
\end{eqnarray}
Noting that $s \leq \min \{K, m/N\}$ and $N$ is a small constant, the complexity of the gOMP can be expressed as $O(Kmn)$.
%
%
%
\begin{table}[t]
\begin{center}
\caption{Complexity of the gOMP algorithm ($k$-th step)} \label{tab:gOMP_cost}
\begin{tabular}{c|c}
\hline Operation & Complexity \\ \hline \hline
Identification                   & $(2m - 1 + N)n - N(N + 1)/2 = O(mn)$ \\
\hline
Estimation of $\hat{\mathbf{x}}_{\Lambda^k}$ & $ \approx 4N^2 km = O(km)$ \\
\hline
Residual update                              & $2Nk m$ \\
\hline \hline
Total                   & $\approx 2mn + (4N^2 + 2N)km = O(mn)$ \\
\hline
\end{tabular}
\end{center}
\end{table}
%
%
%
In practice, however, the iteration number of the gOMP is much smaller than $K$ due to the inclusion of multiple correct indices for each iteration, which saves the complexity of the gOMP substantially.
Indeed, as shown in Fig. \ref{fig:Niter}, the number of iterations
is about $\frac{1}{3}$ of the OMP so that the gOMP has a computational advantage
over the OMP.

\section{Conclusion}
As a cost-effective solution for recovering sparse signals from
compressed measurements, the OMP algorithm has received much
attention in recent years. In this paper, we presented the
generalized version of the OMP algorithm for pursuing efficiency in
reconstructing sparse signals.
Since multiple indices can be identified with no additional
postprocessing operation, the proposed gOMP algorithm lends itself
to parallel-wise processing, which expedites the processing of the
algorithm and thereby reduces the running time. In fact, we
demonstrated in the empirical simulations that the gOMP has excellent
recovery performance comparable to $\ell_1$-minimization technique
with fast processing speed and competitive computational complexity.
We showed from the RIP analysis that if the isometry constant of the sensing matrix satisfies $\delta_{NK} < \frac{\sqrt{N}}{\sqrt{K} + 3 \sqrt{N}}$ then the gOMP algorithm can perfectly recover $K$-sparse signals ($K > 1$) from compressed measurements.
One important point we would like to mention is that the gOMP
algorithm is potentially more effective than what this analysis
tells. Indeed, the bound in \eqref{eq:allsuccess} is derived based
on the worst case scenario where the algorithm selects only one
correct index per iteration (hence requires maximum $K$
iterations). In reality, as observed in the empirical simulations,
it is highly likely that the multiple correct indices are identified for
each iteration and hence the number of iterations is usually much smaller
than that of the OMP.
Therefore, we believe that less strict or probabilistic analysis
will uncover the whole story of the CS recovery performance. Our
future work will be directed towards this avenue.


\appendices

%
%
%
\section{Proof of Lemma \ref{lem:upperbound}} \label{app:upperbound}
\begin{proof}
Let $ w_i $ be the index of the $i$-th largest correlation in magnitude between $\mathbf{r}^k$ and $\{ \varphi_j \}_{j \in F}$ (i.e., columns corresponding to remaining incorrect indices). Also, define the set of indices $W = \left\{ {{w_1},{w_2}, \cdots,{w_N}} \right\}$.
The $\ell_2$-norm of the correlation ${\mathbf{\Phi }}_W' {{\mathbf{r}}^k}$ is expressed as
\begin{eqnarray}
{\left\| {{\mathbf{\Phi }}_W' {{\mathbf{r}}^k}} \right\|_2} &=&
{\left\| {{\mathbf{\Phi }}_W' {\mathbf{P}}_{{\Lambda ^k}}^ \bot
{{\mathbf{\Phi }}_{T - {\Lambda ^k}}}{{\mathbf{x}}_{T - {\Lambda^k}}}} \right\|_2}  \nonumber  \\
&=&
\lefteqn{ \| {\mathbf{\Phi }}_W' {{\mathbf{\Phi }}_{T - {\Lambda
^k}}}{{\mathbf{x}}_{T - {\Lambda ^k}}} } \nonumber \\
& & - {\mathbf{\Phi }}_W' {{\mathbf{P}}_{{\Lambda ^k}}}{{\mathbf{\Phi }}_{T - {\Lambda^k}}}{{\mathbf{x}}_{T - {\Lambda ^k}}} \|_2  \nonumber \\
&\leq& \lefteqn{ \left\| {{\mathbf{\Phi }}_W' {{\mathbf{\Phi }}_{T -
{\Lambda ^k}}}{{\mathbf{x}}_{T - {\Lambda ^k}}}} \right\|_2 } \nonumber \\
& & + {\left\| {{\mathbf{\Phi }}_W' {{\mathbf{P}}_{{\Lambda
^k}}}{{\mathbf{\Phi }}_{T - {\Lambda ^k}}}{{\mathbf{x}}_{T -
{\Lambda ^k}}}} \right\|_2} \label{eq:54}
\end{eqnarray}
where ${\mathbf{P}}_{{\Lambda^k}}^ \bot  = {\mathbf{I}} - {{\mathbf{P}}_{{\Lambda^k}}}$.
Since $W$ and $T - \Lambda^k$ are disjoint (i.e., $W \cap (T -
{\Lambda ^k}) = \emptyset $) and $|W| + |T - \Lambda^k| = N + K -
l$ (note that the number of correct indices in $\Lambda^k$ is $l$ by hypothesis).
Using this together with Lemma \ref{lem:correlationrip}, we have
\begin{eqnarray}
{\left\| {{\mathbf{\Phi }}_W' {{\mathbf{\Phi }}_{T - {\Lambda
^k}}}{{\mathbf{x}}_{T - {\Lambda ^k}}}} \right\|_2} \leq  {\delta
_{N + K - l}}{\left\| {{{\mathbf{x}}_{T - {\Lambda ^k}}}}
\right\|_2}.
\label{eq:j1}
\end{eqnarray}
Similarly, noting that $W \cap {\Lambda ^k} = \emptyset$ and  $|W| + |{\Lambda ^k} |=N + Nk$, we have
\begin{eqnarray}
\lefteqn{ {\left\| {{\mathbf{\Phi }}_W' {{\mathbf{P}}_{{\Lambda
^k}}}{{\mathbf{\Phi }}_{T - {\Lambda ^k}}}{{\mathbf{x}}_{T -
{\Lambda ^k}}}} \right\|_2} } \nonumber \\
& & \leq {\delta _{N + Nk}}{\left\|
{{\mathbf{\Phi }}_{{\Lambda ^k}}^\dag {{\mathbf{\Phi }}_{T -
{\Lambda ^k}}}{{\mathbf{x}}_{T - {\Lambda ^k}}}} \right\|_2}
\label{eq:j2}
\end{eqnarray}
where
\begin{eqnarray} \label{eq:ji}
\lefteqn{ {\left\| {{\mathbf{\Phi }}_{{\Lambda ^k}}^\dag {{\mathbf{\Phi }}_{T
- {\Lambda ^k}}}{{\mathbf{x}}_{T - {\Lambda ^k}}}} \right\|_2} } \nonumber \\
&=&
{\left\| {{{\left( {{\mathbf{\Phi }}_{{\Lambda ^k}}'{{\mathbf{\Phi
}}_{{\Lambda ^k}}}} \right)}^{ - 1}}{\mathbf{\Phi }}_{{\Lambda
^k}}'{{\mathbf{\Phi }}_{T - {\Lambda ^k}}}{{\mathbf{x}}_{T -
{\Lambda ^k}}}} \right\|_2}  \\
& \leq & \label{eq:ghg1} \frac{1}{{1 -  \delta _{Nk} }}  {\left\| {{\mathbf{\Phi }}_{{\Lambda ^k}}'{{\mathbf{\Phi }}_{T - {\Lambda ^k}}}{{\mathbf{x}}_{T - {\Lambda ^k}}}} \right\|_2}   \\
& \leq & \label{eq:ghg2} \frac{{\delta _{Nk + K - l}}}{{1 - \delta
_{Nk} }} {\left\| {{\mathbf{x}}_{T - {\Lambda ^k}}} \right\|_2},
\end{eqnarray}
where (\ref{eq:ghg1}) and (\ref{eq:ghg2}) follow from Lemma \ref{lem:rips} and Lemma \ref{lem:correlationrip}, respectively.
Since ${\Lambda ^k}$ and $T - {\Lambda ^k}$ are disjoint, if the number of correct indices in ${\Lambda^k}$ is $l$, then $\left| {{\Lambda ^k} \cup \left( {T - {\Lambda ^k}} \right)} \right|
= Nk + K - l$.

Using (\ref{eq:54}), (\ref{eq:j1}), (\ref{eq:j2}), and (\ref{eq:ghg2}), we have
\begin{eqnarray}
{\left\| {{\mathbf{\Phi }}_W' {{\mathbf{r}}^k}} \right\|_2} \,\,\,\,\,\,\,\,\,\,\,\,\,\,\,\,\,\,\,\,\,\,\,\,\,\,\,\,\,\,\,\,\,\,\,\,\,\,\,\,\,\,\,\,\,\,\,\,\,\,\,\,\,\,\,\,\,\,\,\,\,\,\,\,\,\,\,\,\,\,\,\,\,\,\,\,\,\,\,\,\,\,\,\,\,\,\,\,\,\,\,
\nonumber \\
\leq
\left( {{\delta _{N + K - l}} + \frac{{{\delta _{N + Nk}}{\delta
_{Nk + K - l}}}} {{1 - {\delta _{Nk}}}}} \right){\left\|
{{{\mathbf{x}}_{T - {\Lambda ^k}}}} \right\|_2}.  \label{eq:left}
\end{eqnarray}

Since $\alpha_i = |\langle \varphi_{w_i}, \mathbf{r}^k \rangle|$, we have
$\|{{\mathbf{\Phi }}_W' {{\mathbf{r}}^k}}\|_1 = \sum\limits_{i = 1}^N { \alpha_i}$.
Now, using the norm inequality\footnote{$\left\| \mathbf{u}
\right\|_1 \leq \sqrt{\left\| \mathbf{u} \right\|_0} \left\|
\mathbf{u} \right\|_2$.}, we have
\begin{eqnarray} \label{eq:dontcare}
{\left\| {{\mathbf{\Phi }}_W' {{\mathbf{r}}^k}} \right\|_2}   \geq
   \frac{1}{{\sqrt N }}\sum\limits_{i = 1}^N { \alpha_i  }.
\end{eqnarray}
Since  $ \alpha_1  \geq  \alpha_2  \geq \cdots \geq  \alpha_N $, it is clear that
${\left\| {{\mathbf{\Phi }}_W' {{\mathbf{r}}^k}} \right\|_2}  \geq
\frac{1} {{\sqrt N }} N \alpha_N   = \sqrt N   \alpha_N$.
Hence, we have
\begin{eqnarray}
\left( {{\delta _{N + K - l}} + \frac{{{\delta _{N +
Nk}}{\delta _{Nk + K - l}}}} {{1 - {\delta _{Nk}}}}}
\right)  {\left\| {{{\mathbf{x}}_{T - {\Lambda
^k}}}} \right\|_2} \nonumber \\
\geq \sqrt N \alpha_N,
\end{eqnarray}
and
\begin{eqnarray}
\alpha_N  \leq \left( {{\delta _{N + K - l}} + \frac{{{\delta _{N +
Nk}}{\delta _{Nk + K - l}}}} {{1 - {\delta _{Nk}}}}}
\right) \frac{{{{\left\| {{{\mathbf{x}}_{T - {\Lambda ^k}}}}
\right\|}_2}}} {{\sqrt N }}. \label{eq:uppercw}
\end{eqnarray}
\end{proof}

%
%
%
%
\section{Proof of Lemma \ref{lem:lowerbound}} \label{app:lowerbound}
\begin{proof}
Since $\beta_1$ is the largest correlation in magnitude between $\mathbf{r}^k$ and $\{\varphi_j\}_{j \in   T - \Lambda^k}$, it is clear that
\begin{eqnarray}
\beta_1 \geq \left| {\left\langle {{\varphi _j},{{\mathbf{r}}^k}} \right\rangle
} \right|
\end{eqnarray}
for all $j \in T - \Lambda^k$, and hence
\begin{eqnarray}
  {\beta _1}  &\geq&  \frac{1}
{{\sqrt {K - l} }}{\left\| {{{{\mathbf{\Phi '}}}_{T - {\Lambda ^k}}}{{\mathbf{r}}^k}} \right\|_2} \\
& = & \frac{1}
{{\sqrt {K - l} }}{\left\| {{{{\mathbf{\Phi '}}}_{T - {\Lambda ^k}}}{\mathbf{P}}_{{\Lambda ^k}}^ \bot  {{\mathbf{\Phi x}}} } \right\|_2} \label{eq:ggeessaa}
\end{eqnarray}
where \eqref{eq:ggeessaa} follows from ${{\mathbf{r}}^k} = \mathbf{y} - \mathbf{\Phi}_{\Lambda^{k}} \mathbf{\Phi}_{\Lambda^{k}}^\dag \mathbf{y} = {\mathbf{P}}_{{\Lambda ^k}}^ \bot {\mathbf{\Phi x}}$.
Using the triangle inequality,
\begin{eqnarray}
 {\beta _1}  &\geq&  \frac{ \left\| {{{{\mathbf{\Phi '}}}_{T - {\Lambda ^k}}}{\mathbf{P}}_{{\Lambda ^k}}^ \bot {{\mathbf{\Phi }}_{T - {\Lambda ^k}}}{{\mathbf{x}}_{T - {\Lambda ^k}}}} \right\|_2 }
{{\sqrt {K - l} }} \\
&\geq& \lefteqn{ \frac{1}{ \sqrt{K - l} }  ( {{\left\| {{{\mathbf{\Phi '}}}_{T - {\Lambda ^k}}}{{\mathbf{\Phi }}_{T - {\Lambda ^k}}}{\mathbf{x}}_{T - {\Lambda ^k}} \right\|}_2} } \nonumber \\
& &- {{\left\| {{{{\mathbf{\Phi '}}}_{T - {\Lambda ^k}}}{{\mathbf{P}}_{{\Lambda ^k}}}{{\mathbf{\Phi }}_{T - {\Lambda ^k}}}{{\mathbf{x}}_{T - {\Lambda ^k}}}} \right\|}_2} ). \label{eq:ggeeewwwq}
\end{eqnarray}
Since $|T - {\Lambda ^k}| = K - l$,
\begin{eqnarray} \label{eq:geaig1aa}
\lefteqn{\left\| {{{{\mathbf{\Phi '}}}_{T - {\Lambda ^k}}}{{\mathbf{\Phi }}_{T - {\Lambda ^k}}}{{\mathbf{x}}_{T - {\Lambda ^k}}}} \right\|_2} \nonumber \\
& & \geq  \left( {1 - {\delta _{K - l}}} \right){\left\| {{{\mathbf{x}}_{T - {\Lambda ^k}}}} \right\|_2},
\end{eqnarray}
and
\begin{eqnarray}
  \lefteqn{\left\| {{{{\mathbf{\Phi '}}}_{T - {\Lambda ^k}}}{{\mathbf{P}}_{{\Lambda ^k}}}{{\mathbf{\Phi }}_{T - {\Lambda ^k}}}{{\mathbf{x}}_{T - {\Lambda ^k}}}} \right\|_2} \nonumber \\
  &\leq&  {\left\| {{{{\mathbf{\Phi '}}}_{T - {\Lambda ^k}}}} \right\|_2}{\left\| {{{\mathbf{P}}_{{\Lambda ^k}}}{{\mathbf{\Phi }}_{T - {\Lambda ^k}}}{{\mathbf{x}}_{T - {\Lambda ^k}}}} \right\|_2}   \\
\label{eq:llls233111} &\leq&  \sqrt {1 + {\delta _{K - l}}} {\left\| {{{\mathbf{P}}_{{\Lambda ^k}}}{{\mathbf{\Phi }}_{T - {\Lambda ^k}}}{{\mathbf{x}}_{T - {\Lambda ^k}}}} \right\|_2}
\end{eqnarray}
where \eqref{eq:llls233111} is from $${\left\| {{{{\mathbf{\Phi '}}}_{T - {\Lambda ^k}}}} \right\|_2} \leq  \sqrt {{\lambda _{\max }}\left( {{{{\mathbf{\Phi '}}}_{T - {\Lambda ^k}}}{{\mathbf{\Phi }}_{T - {\Lambda ^k}}}} \right)}  \leq  \sqrt {1 + {\delta _{K - l}}}. $$ 
Furthermore, we observe that
\begin{eqnarray}  \label{eq:33}
\lefteqn{ \left\| {{{\mathbf{P}}_{{\Lambda ^k}}}{{\mathbf{\Phi }}_{T -
{\Lambda ^k}}}{{\mathbf{x}}_{T - {\Lambda ^k}}}} \right\|_2 } \nonumber \\
& & = \left\| {{{\mathbf{\Phi }}_{{\Lambda ^k}}}{{\left( {{\mathbf{\Phi }}_{{\Lambda ^k}}'{{\mathbf{\Phi }}_{{\Lambda ^k}}}} \right)}^{ - 1}}{\mathbf{\Phi }}_{{\Lambda ^k}}'{{\mathbf{\Phi }}_{T - {\Lambda ^k}}}{{\mathbf{x}}_{T - {\Lambda ^k}}}} \right\|_2  \\
& & \leq  \label{eq:rrrr1}  \lefteqn{\sqrt{{1 + {\delta _{Nk}}}}} \nonumber \\ & & \times \left\| {{{\left( {{\mathbf{\Phi }}_{{\Lambda ^k}}'{{\mathbf{\Phi }}_{{\Lambda ^k}}}} \right)}^{ - 1}}{\mathbf{\Phi }}_{{\Lambda ^k}}'{{\mathbf{\Phi }}_{T - {\Lambda ^k}}}{{\mathbf{x}}_{T - {\Lambda ^k}}}} \right\|_2 \\
& & \leq \label{eq:rrrr2} \frac{\sqrt{{1 + {\delta _{Nk}}}}}
{{{{  {1 - {\delta _{Nk}}}  } }}}\left\| {{\mathbf{\Phi }}_{{\Lambda ^k}}'{{\mathbf{\Phi }}_{T - {\Lambda ^k}}}{{\mathbf{x}}_{T - {\Lambda ^k}}}} \right\|_2   \\
& & \leq  \label{eq:rrrr3} \frac{{\delta _{Nk + K - l} }   \sqrt{1 +
\delta_{Nk}} } { 1 - {\delta _{Nk}} }  \left\| {{\mathbf{x}}_{T -
{\Lambda ^k}}} \right\|_2,
\end{eqnarray}
where (\ref{eq:rrrr1}) is from the definition of RIP and
(\ref{eq:rrrr2}) and (\ref{eq:rrrr3}) follow from Lemma \ref{lem:rips} and \ref{lem:correlationrip}, respectively (${\Lambda^k}$ and $T - \Lambda^k$ are disjoint sets and $\left| {{\Lambda ^k}
\cup \left( {T - {\Lambda ^k}} \right)} \right| = Nk + K - l$).
By combining \eqref{eq:llls233111} and \eqref{eq:rrrr3}, we obtain
\begin{eqnarray}
  \lefteqn{\left\| {{{{\mathbf{\Phi '}}}_{T - {\Lambda ^k}}}{{\mathbf{P}}_{{\Lambda ^k}}}{{\mathbf{\Phi }}_{T - {\Lambda ^k}}}{{\mathbf{x}}_{T - {\Lambda ^k}}}} \right\|_2} \nonumber \\
  & \leq & \!\!\!\! \frac{  \sqrt {1 \!+\! {\delta _{K - l}}}  \sqrt{1 \!+\! \delta_{Nk}} {\delta _{Nk + K - l} }} { 1 - {\delta _{Nk}} }  \left\| {{\mathbf{x}}_{T - {\Lambda ^k}}} \right\|_2 \!.\label{eq:jiafenhaha}
\end{eqnarray}

Finally, by combining \eqref{eq:ggeeewwwq} \eqref{eq:geaig1aa}, and \eqref{eq:jiafenhaha}, we obtain
\begin{eqnarray}
\lefteqn{ \beta_1 \geq \frac{{{{\left\| {{{\mathbf{x}}_{T - {\Lambda ^k}}}} \right\|}_2}}}{{\sqrt {K - l} }} } \nonumber \\
 & \times & \!\!\!\!\!\!\! \left( \!1 \!-\! {\delta _{K - l}} \!-\! \frac{ \sqrt {1 \!+\! {\delta _{K - l}}} \sqrt {1 \!+\! \delta_{Nk}} \delta_{Nk + K - l} }{ 1 - \delta _{Nk} } \! \right)\!\! .
\end{eqnarray}

\end{proof}

%
%
%
%
%
%
%
%
%
%
%
%
\section{Proof of Theorem \ref{thm:errorbound1}} \label{app:errorbound1}
\begin{proof}
We observe that
\begin{eqnarray}
\label{eq:ripggood1}{\left\| {{\mathbf{x}} - {\mathbf{\hat x}}} \right\|_2} &\leq&  \frac{ \left\| {{\mathbf{\Phi }}\left( {{\mathbf{x}} - {\mathbf{\hat x}}} \right)} \right\|_2 }
{{\sqrt {1 - {\delta _{\left| {{\Lambda ^s} \cup T} \right|}}} }} \\
\label{eq:ripggood2} &\leq&  \frac{ \left\| {{\mathbf{\Phi }}\left( {{\mathbf{x}} - {\mathbf{\hat x}}} \right)} \right\|_2 }
{{\sqrt {1 - {\delta _{Ns + K}}} }}\\
\label{eq:ripggood3} &\leq&  \frac{ \left\| {{\mathbf{\Phi x}} - {{\mathbf{\Phi }}_{{\Lambda ^s}}}{\mathbf{\Phi }}_{{\Lambda ^s}}^\dag {\mathbf{y}}} \right\|_2 }
{{\sqrt {1 - {\delta _{Ns + K}}} }} \\
\label{eq:ripggood4} & = & \frac{ \left\| {{\mathbf{y}} - {\mathbf{v}} - {{\mathbf{\Phi
}}_{{\Lambda ^s}}}{\mathbf{\Phi }}_{{\Lambda ^s}}^\dag {\mathbf{y}}}
\right\|_2 } {{\sqrt {1 - {\delta _{Ns + K}}}
}}
\end{eqnarray}
where \eqref{eq:ripggood1} is due to the definition of the RIP,
\eqref{eq:ripggood2} follows from the fact that ${{\mathbf{x}} -
{\mathbf{\hat x}}}$ is at most $(Ns + K)$-sparse ($\delta _{\left|
{{\Lambda ^s} \cup T} \right|} \leq \delta _{Ns + K}$), and
\eqref{eq:ripggood4} is from $\mathbf{\hat x}_{{\Lambda ^s}} =
{\mathbf{\Phi }}_{{\Lambda ^s}}^\dag {\mathbf{y}}$.

Since ${\mathbf{y}} - {{\mathbf{\Phi }}_{{\Lambda ^s}}}{\mathbf{\Phi }}_{{\Lambda ^s}}^\dag {\mathbf{y}} = {\mathbf{y}} - {{\mathbf{P}}_{{\Lambda ^s}}}{\mathbf{y}} = {\mathbf{P}}_{{\Lambda ^s}}^ \bot {\mathbf{y}} = {{\mathbf{r}}^s}$, we further have
\begin{eqnarray}
  {\left\| {{\mathbf{x}} - {\mathbf{\hat x}}} \right\|_2} \!\! &\leq& \!\!\frac{1}
{{\sqrt {1 - {\delta _{Ns + K}}} }}{\left\| {{{\mathbf{r}}^s} - {\mathbf{v}}} \right\|_2}  \\
&\leq& \!\! \frac{1}
{{\sqrt {1 - {\delta _{Ns + K}}} }}\left( {{{\left\| {{{\mathbf{r}}^s}} \right\|}_2} + {{\left\| {\mathbf{v}} \right\|}_2}} \right)  \\
&\leq& \!\! \frac{\epsilon } {{\sqrt {1 - {\delta _{NK}}} }} +
\frac{{{{\left\| {\mathbf{v}} \right\|}_2}}} {{\sqrt {1 - {\delta
_{NK}}} }}
\end{eqnarray}
where the last inequality is due to $Ns + K \leq NK$ (and hence ${\delta _{Ns + K}} \leq  {\delta _{NK}}$) and ${\left\| {{{\mathbf{r}}^s}} \right\|}_2 < \epsilon$.
\end{proof}

%
%
%
%
\section{Proof of Lemma \ref{lem:errorbound1_two}} \label{app:errorbound1_two}
\begin{proof}
Using a triangle inequality,
\begin{eqnarray} \label{eq:jigeren}
    \label{eq:ripggo2od1} {\left\| {{{{\mathbf{\Phi '}}}_W}{{\mathbf{r}}^k}} \right\|_2} & = & {\left\| {{{{\mathbf{\Phi '}}}_W}{\mathbf{P}}_{{\Lambda ^k}}^ \bot \left( {{\mathbf{\Phi x}} + {\mathbf{v}}} \right)} \right\|_2} \\
 \label{eq:ripggo2od2}&\leq& {\left\| {{{{\mathbf{\Phi '}}}_W}{\mathbf{P}}_{{\Lambda ^k}}^ \bot {{\mathbf{\Phi }}_{T - {\Lambda ^k}}}{{\mathbf{x}}_{T - {\Lambda ^k}}}} \right\|_2} \nonumber \\
 & & + {\left\| {{{{\mathbf{\Phi '}}}_W}{\mathbf{P}}_{{\Lambda ^k}}^ \bot {\mathbf{v}}} \right\|_2}  \\
\label{eq:ripggo2od3}  &\leq&  {\left\| {{{{\mathbf{\Phi
'}}}_W}{{\mathbf{\Phi }}_{T - {\Lambda ^k}}}{{\mathbf{x}}_{T -
{\Lambda ^k}}}} \right\|_2} \nonumber \\
& & + {\left\| {{{{\mathbf{\Phi
'}}}_W}{{\mathbf{P}}_{{\Lambda ^k}}}{{\mathbf{\Phi }}_{T - {\Lambda
^k}}}{{\mathbf{x}}_{T - {\Lambda ^k}}}} \right\|_2} \nonumber \\
& & + {\left\|
{{{{\mathbf{\Phi '}}}_W}{\mathbf{P}}_{{\Lambda ^k}}^ \bot
{\mathbf{v}}} \right\|_2}.
\end{eqnarray}
Using (\ref{eq:j1}), (\ref{eq:j2}), and (\ref{eq:ghg2}), ${\left\| {{{{\mathbf{\Phi '}}}_W}{{\mathbf{r}}^k}} \right\|_2}$ is upper bounded by
\begin{eqnarray} \label{eq:JSJSG}
\lefteqn{\left\| {{{{\mathbf{\Phi '}}}_W}{{\mathbf{r}}^k}} \right\|_2} \nonumber \\
& & \leq
\left( {{\delta _{N + K - l}} + \frac{{{\delta _{N + Nk}}{\delta
_{Nk + K - l}}}} {{1 - {\delta _{Nk}}}}} \right){\left\|
{{{\mathbf{x}}_{T - {\Lambda ^k}}}} \right\|_2} \nonumber \\
& & + {\left\|
{{{{\mathbf{\Phi '}}}_W}{\mathbf{P}}_{{\Lambda ^k}}^ \bot
{\mathbf{v}}} \right\|_2}.
\end{eqnarray}
Further, we have
\begin{eqnarray}
 \label{eq:lll1} {\left\| {{{{\mathbf{\Phi '}}}_W}{\mathbf{P}}_{{\Lambda ^k}}^ \bot {\mathbf{v}}} \right\|_2} \!\!\!\! &\leq& \!\!\!\! {\left\| {{{{\mathbf{\Phi '}}}_W}} \right\|_2}{\left\| {{\mathbf{P}}_{{\Lambda ^k}}^ \bot {\mathbf{v}}} \right\|_2}  \\
\label{eq:lll2} &\leq& \!\!\!\! \sqrt {{\lambda _{\max }}\left( {{{{\mathbf{\Phi '}}}_W}{{\mathbf{\Phi }}_W}} \right)} \nonumber \\
& & \times {\left\| {{\mathbf{P}}_{{\Lambda ^k}}^ \bot {\mathbf{v}}} \right\|_2}  \\
\label{eq:lll3}&\leq& \!\!\!\! \sqrt {1 + {\delta _N}} {\left\| {{\mathbf{P}}_{{\Lambda ^k}}^ \bot {\mathbf{v}}} \right\|_2} \\
\label{eq:lll4} &\leq& \!\!\!\! \sqrt {1 + {\delta _N}} {\left\|
{\mathbf{v}} \right\|_2}.
\end{eqnarray}
Plugging \eqref{eq:lll4} into \eqref{eq:JSJSG}, we have
\begin{eqnarray} \label{eq:rwwwwight}
  \lefteqn{\left\| {{{{\mathbf{\Phi '}}}_W}{{\mathbf{r}}^k}} \right\|_2} \nonumber \\
  & & \leq \left( {{\delta _{N + K - l}} + \frac{{{\delta _{N + Nk}}{\delta _{Nk + K - l}}}}
{{1 - {\delta _{Nk}}}}} \right){\left\| {{{\mathbf{x}}_{T - {\Lambda
^k}}}} \right\|_2} \nonumber \\
& & + \sqrt {1 + {\delta _N}} {\left\| {\mathbf{v}}
\right\|_2}.
\end{eqnarray}
Using ${\left\| {{\mathbf{\Phi }}_W' {{\mathbf{r}}^k}} \right\|_2}  \geq \sqrt N \alpha_N$ and \eqref{eq:rwwwwight}, we have
\begin{eqnarray}
\alpha _N & \leq &  \left( {\delta _{N + K - l}} + \frac{{{\delta _{N + Nk}}{\delta _{Nk + K - l}}}}{1 - {\delta _{Nk}}} \right)\frac{{{{\left\| {{{\mathbf{x}}_{T - {\Lambda ^k}}}} \right\|}_2}}} {{\sqrt N }} \nonumber \\
 & & + \frac{{\sqrt {1 + {\delta _N}} {{\left\| {\mathbf{v}} \right\|}_2}}}{{\sqrt N }}.
\end{eqnarray}

Next, we derive a lower bound for $\beta_1$. First, recalling that ${\beta _1} = \mathop {\max }\limits_{j:j \in T - {\Lambda ^k}} \left| {\left\langle {{\varphi _j},{{\mathbf{r}}^k}} \right\rangle } \right|$, we have
\begin{eqnarray}
  {\beta _1} \!\!\! &\geq& \!\!\! \frac{1}
{{\sqrt {K - l} }}{\left\| {{{{\mathbf{\Phi '}}}_{T - {\Lambda ^k}}}{{\mathbf{r}}^k}} \right\|_2} \\
& = & \!\!\! \frac{1} {{\sqrt {K - l} }}{\left\| {{{{\mathbf{\Phi '}}}_{T -
{\Lambda ^k}}}{\mathbf{P}}_{{\Lambda ^k}}^ \bot \left(
{{\mathbf{\Phi x}} + {\mathbf{v}}} \right)} \right\|_2}.
\end{eqnarray}
Using a triangle inequality,
\begin{eqnarray}
 {\beta _1} \!\!\! &\geq& \!\!\! \frac{\left\| {{{{\mathbf{\Phi '}}}_{T - {\Lambda ^k}}}{\mathbf{P}}_{{\Lambda ^k}}^ \bot {{\mathbf{\Phi }}_{T - {\Lambda ^k}}}{{\mathbf{x}}_{T - {\Lambda ^k}}}} \right\|_2}{{\sqrt {K - l} }} \nonumber \\
& & - \frac{\left\| {{{{\mathbf{\Phi '}}}_{T - {\Lambda ^k}}}{\mathbf{P}}_{{\Lambda ^k}}^ \bot {\mathbf{v}}} \right\|_2}{{\sqrt {K - l} }}   \\
&\geq& \!\!\! \frac{{{\left\| {{{{\mathbf{\Phi '}}}_{T - {\Lambda ^k}}}{{\mathbf{\Phi }}_{T - {\Lambda ^k}}}{{\mathbf{x}}_{T - {\Lambda ^k}}}} \right\|}_2}} {{\sqrt {K - l} }} \nonumber \\
& & - \frac{{\left\| {{{{\mathbf{\Phi '}}}_{T - {\Lambda^k}}}{{\mathbf{P}}_{{\Lambda ^k}}}{{\mathbf{\Phi }}_{T - {\Lambda^k}}}{{\mathbf{x}}_{T - {\Lambda ^k}}}} \right\|}_2}{{\sqrt {K - l} }} \nonumber \\
& & - \frac{{\left\|{{{{\mathbf{\Phi '}}}_{T - {\Lambda ^k}}}{\mathbf{P}}_{{\Lambda^k}}^ \bot {\mathbf{v}}} \right\|}_2}{{\sqrt {K - l} }} . \label{eq:111llls3sssss}
\end{eqnarray}
From \eqref{eq:geaig1aa} and \eqref{eq:jiafenhaha}, we have
${\left\| {{{{\mathbf{\Phi '}}}_{T - {\Lambda ^k}}}{{\mathbf{\Phi }}_{T - {\Lambda ^k}}}{{\mathbf{x}}_{T - {\Lambda ^k}}}} \right\|_2} \geq  \left( {1 - {\delta _{K - l}}} \right){\left\| {{{\mathbf{x}}_{T - {\Lambda ^k}}}} \right\|_2}$,
and
\begin{eqnarray} \label{eq:llls3sssss}
  \lefteqn{\left\| {{{{\mathbf{\Phi '}}}_{T - {\Lambda ^k}}}{{\mathbf{P}}_{{\Lambda ^k}}}{{\mathbf{\Phi }}_{T - {\Lambda ^k}}}{{\mathbf{x}}_{T - {\Lambda ^k}}}} \right\|_2} \nonumber \\
  & & \leq \!\! \frac{{\sqrt {1 \!+\! {\delta _{K \!-\! l}}} \! \sqrt {1 \!+\! {\delta _{Nk}}} {\delta _{Nk + K - l}}}}
{{1 - {\delta _{Nk}}}}{\left\| \! {{{\mathbf{x}}_{T \!-\! {\Lambda ^k}}}} \! \right\|_2}.
\end{eqnarray}
Also,
\begin{eqnarray}
  {\left\| {{{{\mathbf{\Phi '}}}_{T - {\Lambda ^k}}}{\mathbf{P}}_{{\Lambda ^k}}^ \bot {\mathbf{v}}} \right\|_2} \!\!\! &\leq& \!\!\! {\left\| {{{{\mathbf{\Phi '}}}_{T - {\Lambda ^k}}}} \right\|_2}{\left\| {{\mathbf{P}}_{{\Lambda ^k}}^ \bot {\mathbf{v}}} \right\|_2}  \\
&\leq& \!\!\! \sqrt {1 + {\delta _{K - l}}} {\left\| {{\mathbf{P}}_{{\Lambda ^k}}^ \bot {\mathbf{v}}} \right\|_2} \\
\label{eq:geaig3}  &\leq& \!\!\! \sqrt {1 + {\delta _{K - l}}} {\left\|
{\mathbf{v}} \right\|_2}.
\end{eqnarray}
Finally, by combining \eqref{eq:111llls3sssss}, \eqref{eq:llls3sssss} and \eqref{eq:geaig3}, we obtain
\begin{eqnarray}
{\beta _1} &\geq & \left(\!\! {1 - {\delta _{K - l}} - \frac{{\sqrt {1 +
{\delta _{K - l}}} \sqrt {1 + {\delta _{Nk}}} {\delta _{Nk + K -
l}}}} {{1 - {\delta _{Nk}}}}} \!\! \right) \nonumber \\
& & \times \frac{{{{\left\|
{{{\mathbf{x}}_{T - {\Lambda ^k}}}} \right\|}_2}}} {{\sqrt {K - l}
}} 
- \frac{{\sqrt {1 + {\delta _{K - l}}} {{\left\| {\mathbf{v}}
\right\|}_2}}} {{\sqrt {K - l} }}.
\end{eqnarray}
\end{proof}

\section{Proof of Lemma \ref{lem:errorbound1_two1}} \label{app:errorbound1_two1}

\begin{proof}
Using the relaxations of the isometry constants in (\ref{eq:monoto}),
we have
\begin{eqnarray}
  {\alpha _N} &\leq&  \left( {{\delta _{N + K - l}} + \frac{{{\delta _{N + Nk}}{\delta _{Nk + K - l}}}}
{{1 - {\delta _{Nk}}}}} \right) \nonumber   \\
& & \times \frac{{{{\left\| {{{\mathbf{x}}_{T - {\Lambda ^k}}}} \right\|}_2}}} {{\sqrt N }} + \frac{{\sqrt {1 +{\delta _N}} {{\left\| {\mathbf{v}} \right\|}_2}}}
{{\sqrt N }}    \nonumber   \\
&\leq&  \left( {{\delta _{NK}} + \frac{{{\delta _{NK}}{\delta
_{NK}}}} {{1 - {\delta _{NK}}}}} \right)\frac{{{{\left\|
{{{\mathbf{x}}_{T - {\Lambda ^k}}}} \right\|}_2}}} {{\sqrt N }} \nonumber   \nonumber    \\
& & +
\frac{{\sqrt {1 + {\delta _N}} {{\left\| {\mathbf{v}} \right\|}_2}}}
{{\sqrt N }}    \nonumber   \\
& = & \frac{{{\delta _{NK}}}} {{1 - {\delta _{NK}}}}\frac{{{{\left\|
{{{\mathbf{x}}_{T - {\Lambda ^k}}}} \right\|}_2}}} {{\sqrt N }}
+ \frac{{\sqrt {1 + {\delta _N}} {{\left\| {\mathbf{v}} \right\|}_2}}}
{{\sqrt N }}
\end{eqnarray}
and
\begin{eqnarray}
  {\beta _1} &\geq&  \left( {1 - {\delta _{K - l}} - \frac{{\sqrt {1 + {\delta _{K - l}}} \sqrt {1 + {\delta _{Nk}}} {\delta _{Nk + K - l}}}}
{{1 - {\delta _{Nk}}}}} \right) \nonumber   \\
    & & \times \frac{{{{\left\| {{{\mathbf{x}}_{T -
{\Lambda ^k}}}} \right\|}_2}}} {{\sqrt {K - l} }} - \frac{{\sqrt {1
+ {\delta _{K - l}}} {{\left\| {\mathbf{v}} \right\|}_2}}}
{{\sqrt {K - l} }} \nonumber \\
&\geq&  \left( {1 - {\delta _{NK}} - \frac{{\sqrt {1 + {\delta
_{NK}}} \sqrt {1 + {\delta _{NK}}} {\delta _{NK}}}} {{1 - {\delta
_{NK}}}}} \right)   \nonumber   \\
& & \times \frac{{{{\left\| {{{\mathbf{x}}_{T - {\Lambda
^k}}}} \right\|}_2}}} {{\sqrt {K - l} }} - \frac{{\sqrt {1 + {\delta
_{K - l}}} {{\left\| {\mathbf{v}} \right\|}_2}}}
{{\sqrt {K - l} }} \nonumber \\
&=& \frac{{1 - 3{\delta _{NK}}}} {{1 - {\delta
_{NK}}}}\frac{{{{\left\| {{{\mathbf{x}}_{T - {\Lambda ^k}}}}
\right\|}_2}}} {{\sqrt {K - l} }} - \frac{{\sqrt {1 + {\delta _{K -
l}}} {{\left\| {\mathbf{v}} \right\|}_2}}} {{\sqrt {K - l} }}.
\end{eqnarray}

The bounds on $\alpha_N$ and $\beta_1$ imply that a sufficient condition of $\alpha_N < \beta_1$ is
\begin{eqnarray} \frac{{1 - 3{\delta _{NK}}}}
{{1 - {\delta _{NK}}}}\frac{{{{\left\| {{{\mathbf{x}}_{T - {\Lambda
^k}}}} \right\|}_2}}} {{\sqrt {K - l} }} - \frac{{\sqrt {1 + {\delta
_{K - l}}} {{\left\| {\mathbf{v}} \right\|}_2}}} {{\sqrt {K - l} }}
 \nonumber  \\
 > \frac{{{\delta _{NK}}}} {{1 - {\delta _{NK}}}}\frac{{{{\left\|
{{{\mathbf{x}}_{T - {\Lambda ^k}}}} \right\|}_2}}} {{\sqrt N }} +
\frac{{\sqrt {1 + {\delta _N}} {{\left\| {\mathbf{v}} \right\|}_2}}}
{{\sqrt N }}.
\end{eqnarray}
After some manipulations, we have
\begin{eqnarray}
\lefteqn{{\left\| {{\mathbf{x}}_{T-{{\Lambda }^{k}}}} \right\|}_{2} >}  \nonumber   \\
& & \!\!\!\!\!\!\!\!\!\!\!\! \frac{\left(\! \sqrt{1 \!+\! {{\delta}_{K-l}}} \!+\!\! \sqrt{\frac{K}{N} \!+\! \frac{K-l}{N} {{\delta }_{N}}} \! \right) \!\! \left( 1 \!-\! {{\delta }_{NK}} \right)}{1-3{{\delta
}_{NK}}-\sqrt{\frac{K-l}{N} }{{\delta }_{NK}}} \! {{\| \mathbf{v} \|}_{2}}.
\end{eqnarray}
Since $K-l \leq K$, this condition is guaranteed if
\begin{align} \label{eq:ggaalss1111l}
&{{\left\| {{\mathbf{x}}_{T-{{\Lambda }^{k}}}} \right\|}_{2}}>      \nonumber   \\
&\frac{\left( \sqrt{1+{{\delta
}_{K}}}+\sqrt{\frac{K}{N} \left({1 + {\delta }_{N}}\right) }
\right)\left( 1-{{\delta }_{NK}} \right)}{1-3{{\delta
}_{NK}}-\sqrt{\frac{K}{N} }{{\delta }_{NK}}}{{\left\| \mathbf{v}
\right\|}_{2}}.
\end{align}
\end{proof}

%
%

\section{Computational cost for the ``estimate'' step of gOMP} \label{app:LS}

In the $k$-th iteration, the gOMP estimates the nonzero elements of
$\mathbf{x}$ by solving an LS problem,
\begin{equation}
\label{eq:LS} {{\mathbf{\hat x}}_{\Lambda^{k}}} = \arg \mathop
{\min}\limits_{\mathbf{x}}{\left\|\mathbf{y}-\mathbf{\Phi}_{\Lambda^{k}}\mathbf{x}\right\|}_{2}
= \left(\mathbf{\Phi}_{\Lambda^k}'
\mathbf{\Phi}_{\Lambda^k}\right)^{-1} \mathbf{\Phi}_{\Lambda^k}'
\mathbf{y}.
\end{equation}
To solve \eqref{eq:LS}, we employ the MGS algorithm in which the QR
decomposition of previous iteration is maintained and, therefore,
the computational cost can be reduced.

Without loss of generality, we assume $\mathbf{\Phi}_{\Lambda^k} = {\left(
\begin{matrix}
\varphi_1 & \varphi_2 & \cdots &
\varphi_{Nk}\\
\end{matrix} \right)}$.
The QR decomposition of $\mathbf{\Phi}_{\Lambda^k}$ is given by
\begin{eqnarray*}
\mathbf{\Phi}_{\Lambda^k} = \mathbf{Q R}
\end{eqnarray*}
where $\mathbf{Q} = {\left(
\begin{matrix}
\mathbf{q}_1& \mathbf{q}_2& \cdots &
\mathbf{q}_{Nk}\\
\end{matrix} \right)} \in \mathbb{R}^{m \times {Nk}}$ consists of ${Nk}$
orthonormal columns and $\mathbf{R}\in
\mathbb{R}^{{Nk} \times {Nk}}$ is an upper triangular matrix,
$$ \mathbf{R} = {\left(
\begin{matrix}
   \langle \mathbf{q}_1, \varphi_1 \rangle & \langle \mathbf{q}_2, \varphi_2 \rangle & \cdots  & \langle \mathbf{q}_1,  \varphi_{Nk} \rangle  \\
   0   & \langle \mathbf{q}_2, \varphi_2 \rangle & \cdots  & \langle \mathbf{q}_2, \varphi_{Nk} \rangle  \\
       &   & \cdots  &  \\
   0   &  0& \cdots  & \langle \mathbf{q}_{Nk}, \varphi_{Nk} \rangle  \\
\end{matrix} \right)}.$$
For notation simplicity we denote $R_{i,j} = \langle \mathbf{q}_i,
\varphi_j \rangle $ and $p = N(k - 1)$. In addition, we denote the
QR decomposition of the $(k - 1)$-th iteration as
$\mathbf{\Phi}_{\Lambda^{k - 1}} = \mathbf{Q_{-1}} \mathbf{R_{-1}}$.
Then it is clear that
\begin{eqnarray}
\mathbf{Q} = {\left(
\begin{matrix}
\mathbf{Q_{-1}} & \mathbf{Q_0}
\end{matrix} \right)} \hspace{2mm} \mbox{and} \hspace{2mm}
\mathbf{R}  =  {\left(
\begin{matrix}
\mathbf{R_{-1}} & \mathbf{R_a}  \\
\mathbf{0} & \mathbf{R_b}  \\
\end{matrix} \right)}.
\end{eqnarray}
where $\mathbf{Q_0} = {\left(
\begin{matrix}
\mathbf{q}_{p + 1} & \cdots & \mathbf{q}_{Nk}
\end{matrix} \right)} \in \mathbb{R}^{m \times N}$
and $\mathbf{R_a}$ and $\mathbf{R_b}$ are given by
\begin{eqnarray}
\mathbf{R_a} &=& {\left(
\begin{matrix}
R_{1, p + 1}          &\cdots & R_{1, Nk}         \\
           \vdots   &   &     \vdots      \\
 R_{p, p + 1} &\cdots & R_{p, Nk}\\
\end{matrix} \right)}, \nonumber    \\
\mathbf{R_b} &=& {\left(
\begin{matrix}
 R_{p + 1, p + 1} &\cdots & R_{p + 1, Nk}\\
                             &\ddots & \vdots    \\
               \mathbf{0}            &       & R_{Nk, Nk}        \\
\end{matrix} \right)}.
\end{eqnarray}
Applying $\mathbf{\Phi}_{\Lambda^k} = \mathbf{Q R}$ to
\eqref{eq:LS}, we have
\begin{eqnarray}
\label{eq:LS1} {{\mathbf{\hat x}}_{\Lambda^{k}}} = \left(\mathbf{R}'
\mathbf{R}\right)^{-1} \mathbf{R}' \mathbf{Q}' \mathbf{y}.
\end{eqnarray}
We count the cost of solving \eqref{eq:LS1} in the following steps.
Here we assess the computational cost by counting
floating-point operations (flops). That is, each $ +, -, *, /,\sqrt{\hspace{3mm}} $ counts as one flop.

\begin{itemize}
\item \textbf{Cost of QR decomposition}: \\
To obtain $\mathbf{Q}$ and $\mathbf{R}$, one only needs to compute
$\mathbf{Q_0}$, $\mathbf{R_a}$ and $\mathbf{R_b}$ since the previous
data, $\mathbf{Q_{-1}}$ and $ \mathbf{R_{-1}}$ are stored. For $j =
1$ to $N$, we sequentially calculate
\begin{equation*}
\{R_{i, p + j}\}_{i = 1,2,\cdots, p + j - 1} = \{\left\langle
\mathbf{q}_i, \varphi_j
\right\rangle\}_{i = 1,2,\cdots, p + j - 1},
\end{equation*}
\begin{eqnarray*}
\hat{\mathbf{q}}_{p + j} &=& \varphi_{p + j} - \sum_{i = 1}^{p + j -1} R_{i,p + j} \mathbf{q}_{i}, \\
 \mathbf{q}_{p + j} &=& \frac{\hat{\mathbf{q}}_{p + j}}{\| \hat{\mathbf{q}}_{p + j} \|_2},  \\
\\
R_{p + j, p + j} &=& \left\langle \mathbf{q}_{p + j},
\varphi_{p + j} \right\rangle.
\end{eqnarray*}
Taking $j = 1$ for example. One first computes $\{R_{i, p + 1}\}_{i
= 1, 2, \cdots, p}$ using $\mathbf{Q_{-1}}$ (requires $p(2m - 1)$
flops) and then computes $ \hat{\mathbf{q}}_{p + 1} = \varphi_{p +
1} - \sum_{i = 1}^p R_{i,p + 1} \mathbf{q}_{i}$ (requires $2mp$
flops). Then, normalization of $\hat{\mathbf{q}}_{p + 1}$ requires $3m$
flops.
Finally, computing $R_{p + 1, p + 1}$ requires $2m - 1$
flops.
The cost of this example amounts to $4mp + 5m - p - 1$.
Similarly, one can calculate the other data in $\mathbf{Q_0}$ and
${\left(
\begin{matrix}
\mathbf{R_a}  & \mathbf{R_b}  \\
\end{matrix} \right)}'$.
In summary, the cost for this QR factorization becomes
\begin{eqnarray*}
\mathcal{C}_{\mathbf{QR}} =  4{N}^{2}mk - 2m{N}^{2} + 3mN -
{N}^{2}k + \frac{1}{2} {N}^{2} -  \frac{1}{2}N.
\end{eqnarray*}

\item \textbf{Cost of calculating} $\mathbf{Q}' \mathbf{y}$ \\
Since $\mathbf{Q} = {\left(
\begin{matrix}
\mathbf{Q_{-1}} & \mathbf{Q_0}
\end{matrix} \right)}$, we have
\begin{eqnarray*}
\mathbf{Q}' \mathbf{y} = {\left(
\begin{matrix}
\mathbf{Q'_{- 1}} \mathbf{y} \\
\mathbf{Q'_0} \mathbf{y}  \\
\end{matrix} \right)}.
\end{eqnarray*}
By reusing the data $\mathbf{Q}'_{- 1} \mathbf{y}$, $\mathbf{Q}'
\mathbf{y}$ is solved with
\begin{eqnarray*}
\mathcal{C}_{1} = N(2m - 1).
\end{eqnarray*}

\item \textbf{Cost of calculating} $\mathbf{R}' \mathbf{Q}' \mathbf{y}$: \\
Applying $\mathbf{R}'$ to the vector $\mathbf{Q}' \mathbf{y}$, we
have
\begin{eqnarray*}
\mathbf{R}' \mathbf{Q}' \mathbf{y} = {\left(
\begin{matrix}
\mathbf{R'_{-1}} \mathbf{Q'_{- 1}} \mathbf{y} \\
\mathbf{R'_a} \mathbf{Q}'_{- 1} \mathbf{y} + \mathbf{R'_b} \mathbf{Q}'_{0} \mathbf{y}\\
\end{matrix} \right)}.
\end{eqnarray*}
Since the data $\mathbf{R'_{-1}} \mathbf{Q'_{- 1}} \mathbf{y}$ can
be reused, $\mathbf{R}' \mathbf{Q}' \mathbf{y}$ is solved with
\begin{eqnarray*}
\mathcal{C}_{2} = 2N^2k - N^2.
\end{eqnarray*}

\item \textbf{Cost of calculating} $\left(\mathbf{R}'
\mathbf{R}\right)^{-1}$ \\
Since $\mathbf{R}$ is an upper triangular matrix, $
\left(\mathbf{R}' \mathbf{R}\right)^{-1} = (\mathbf{R}')^{-1}
\mathbf{R}^{-1}.$
Using the block matrix inversion formula, we have
\begin{eqnarray*} \label{eq:Rj4}
\mathbf{R}^{-1}  &=& \left(
\begin{matrix}
\mathbf{R_{-1}} & \mathbf{R_a}  \\
\mathbf{0} & \mathbf{R_b}  \\
\end{matrix} \right)^{-1}       \\
&=&
{\left(
\begin{matrix}
(\mathbf{R_{-1}})^{-1} & -(\mathbf{R_{-1}})^{-1} \mathbf{R_a} (\mathbf{R_b})^{-1} \\
\mathbf{0} & (\mathbf{R_b})^{-1} \\
\end{matrix} \right)}.
\end{eqnarray*}
\normalsize Then we calculate $\left(\mathbf{R}'
\mathbf{R}\right)^{-1} = (\mathbf{R}')^{-1} \mathbf{R}^{-1} $, i.e.,
\begin{eqnarray*}
\left(\mathbf{R}' \mathbf{R}\right)^{-1} = {\left(
\begin{matrix}
M_1 & M_2\\
M_3 & M_4 \\
\end{matrix} \right)}
\end{eqnarray*}
where
\begin{eqnarray}
M_1 &=& (\mathbf{R'_{-1}})^{-1}  (\mathbf{R_{-1}})^{-1},\nonumber \\
M_2 &=& - (\mathbf{R'_{-1}})^{-1} (\mathbf{R_{-1}})^{-1} \mathbf{R_a} (\mathbf{R_b})^{-1},\nonumber \\
M_3 &=& -(\mathbf{R'_b})^{-1}  \mathbf{R'_a} (\mathbf{R'_{-1}})^{-1}  (\mathbf{R_{-1}})^{-1}, \nonumber \\
M_4 &=& (\mathbf{R'_b})^{-1}  (\mathbf{R_b})^{-1}. \nonumber
\end{eqnarray}
We can reuse the data $(\mathbf{R'_{-1}})^{-1}
(\mathbf{R_{-1}})^{-1}$ so that the cost of calculating
$(\mathbf{R'_b})^{-1}$, $ (\mathbf{R'_b})^{-1} (\mathbf{R_b})^{-1}$,
and $ -(\mathbf{R'_{-1}})^{-1} (\mathbf{R_{-1}})^{-1} \mathbf{R_a}
(\mathbf{R_b})^{-1}$ becomes $N(N + 1)(2N + 1)/3$ (using Gaussian
elimination method), $N(N + 1)(2N + 1)/6$, and $2 N^3 k^2 - 4N^3 k +
2N^3$, respectively.  The cost for computing $\left(\mathbf{R}'
\mathbf{R}\right)^{-1}$ is
\begin{eqnarray*}
\mathcal{C}_{3} = 2N^3k^2  - 4N^3  k + 3 N^3 + \frac{3}{2} N ^ 2 +
\frac{1}{2} N.
\end{eqnarray*}

\item \textbf{Cost of calculating} $ {{\mathbf{\hat x}}_{\Lambda^{k}}} = \left(\mathbf{R}'
\mathbf{R}\right)^{-1} \mathbf{R}' \mathbf{Q}' \mathbf{y} $ \\
Applying $\left(\mathbf{R}' \mathbf{R}\right)^{-1}$ to the vector
$\mathbf{R}' \mathbf{Q}' \mathbf{y} $, we obtain
\begin{eqnarray*}
{{\mathbf{\hat x}}_{\Lambda^{k}}}  
= {\left(
\begin{matrix}
(\mathbf{R'_{-1}})^{-1}  (\mathbf{R_{-1}})^{-1} \mathbf{R'_{-1}} \mathbf{Q'_{- 1}} \mathbf{y} + \xi_1 \\
\xi_2 + \xi_3 \\
\end{matrix} \right)}
\end{eqnarray*}
where \small
\begin{eqnarray}
\xi_1 & = & - (\mathbf{R'_{-1}})^{-1} (\mathbf{R_{-1}})^{-1}
\mathbf{R_a} (\mathbf{R_b})^{-1} \mathbf{R'_a} \mathbf{Q}'_{- 1}
\mathbf{y} \nonumber \\
    & & + \mathbf{R'_b} \mathbf{Q}'_{0} \mathbf{y}, \nonumber \\
\xi_2 & = & -(\mathbf{R'_b})^{-1}  \mathbf{R'_a}
(\mathbf{R'_{-1}})^{-1}  (\mathbf{R_{-1}})^{-1} \mathbf{R'_{-1}}
\mathbf{Q'_{- 1}} \mathbf{y}, \nonumber \\
\xi_3 & = & (\mathbf{R'_b})^{-1}  (\mathbf{R_b})^{-1} \mathbf{R'_a}
\mathbf{Q}'_{- 1} \mathbf{y} + \mathbf{R'_b} \mathbf{Q}'_{0}
\mathbf{y}. \nonumber
\end{eqnarray}
\normalsize We can reuse $(\mathbf{R'_{-1}})^{-1}
(\mathbf{R_{-1}})^{-1} \mathbf{R'_{-1}} \mathbf{Q'_{- 1}}
\mathbf{y}$ so that the computation of $\xi_1$, $\xi_2$ and $\xi_3$
requires $(2N - 1)N(k - 1)$, $(2N(k - 1) - 1)N$ and $(2N - 1)N$
flops, respectively. The cost of this step becomes
\begin{eqnarray*}
\mathcal{C}_{4} = 4 N^2 k -2 N^2.
\end{eqnarray*}
\end{itemize}

In summary, whole cost of solving LS problem in the $k$-th iteration
of the gOMP is the sum of the above and is given by
\begin{eqnarray*}
\mathcal{C}_{LS} &=&
\mathcal{C}_{QR} + \mathcal{C}_{1} + \mathcal{C}_{2} +
\mathcal{C}_{3} + \mathcal{C}_{4} \\
    &=& 4N^2 km + (- 2 N^2 + 5N )m + 2N^3 k^2  \\
    & & + (-4N^3 + 5 N ^2)k + 3 N^3- N ^2 - N.
\end{eqnarray*}

\section*{Acknowledgment}
The authors would like to thank the anonymous reviewers and for their valuable suggestions that improved the presentation of the paper.

\bibliographystyle{IEEEbib}
\bibliography{CS_refs}

\vspace{-1.0cm}
\begin{biography}{Jian Wang  (Student Member, IEEE) }
received the B.S. degree in Material Engineering and M.S. degree in Information and Communication Engineering from Harbin Institute of Technology, China, in 2006 and 2009, respectively. He is currently working toward the Ph.D. degree in Electrical and Computer Engineering in Korea University. His research interests include compressive sensing, wireless communications, and statistical learning.
\end{biography}

\vspace{-1.0cm}
\begin{biography}{Seokbeop Kwon (Student Member, IEEE) }
received the B.S. and M.S degrees in the School of Information and Communication, Korea University, in 2008 and 2010, where he is currently working toward the Ph.D. degree. His research interests include compressive sensing and signal processing.
\end{biography}

\vspace{-1.0cm}
\begin{biography}{Byonghyo Shim  (Senior Member, IEEE) }
received the B.S. and M.S. degrees in control and instrumentation engineering (currently electrical engineering) from Seoul National University, Korea, in 1995 and 1997, respectively and the M.S. degree in mathematics and the Ph.D. degree in electrical and computer engineering from the University of Illinois at Urbana-Champaign, in 2004 and 2005, respectively.

From 1997 and 2000, he was with the department of electronics engineering at the Korean Air Force Academy as an Officer (First Lieutenant) and an Academic Full-time Instructor. From 2005 to 2007, he was with the Qualcomm Inc., San Diego, CA, as a staff member. Since September 2007, he has been with the school of information and communication, Korea University, where he is currently an associate professor. His research interests include wireless communications, compressive sensing, applied linear algebra, and information theory.

Dr. Shim was the recipient of 2005 M. E. Van Valkenburg research award from ECE department of University of Illinois and 2010 Haedong young engineer award from IEEK.
\end{biography}

\end{document}